\documentclass[a4paper,11pt,reqno,english]{amsart}

\usepackage{amsmath,amssymb,amsthm,mathtools}
\usepackage[a4paper,margin=1in]{geometry}
\usepackage[main=english]{babel}
\usepackage{xcolor}
\usepackage[colorlinks=true,linkcolor=blue,citecolor=blue,urlcolor=blue]{hyperref}

\usepackage[T1]{fontenc}
\usepackage{mathpazo}

\theoremstyle{plain}
    \newtheorem{thm}{Theorem}[section]
    \newtheorem{lem}[thm]{Lemma}
    \newtheorem{prop}[thm]{Proposition}
    \newtheorem{cor}[thm]{Corollary}
    \newtheorem{introthm}{Theorem}

\theoremstyle{definition}
    \newtheorem{defn}[thm]{Definition}

\theoremstyle{remark}
    \newtheorem{rmk}[thm]{Remark}
    \newtheorem{example}[thm]{Example}

\numberwithin{equation}{section}

\newcommand{\RR}{\mathbb{R}}
\newcommand{\CC}{\mathbb{C}}
\newcommand{\ZZ}{\mathbb{Z}}
\newcommand{\ii}{\sqrt{-1}}
\newcommand{\PP}{\mathbb{P}}
\DeclareMathOperator{\Ric}{Ric}

\title{Reduced Solutions of Toda Systems for Arbitrary Complex Simple Lie Algebras}
\date{\today}
\author{Yiqian Shi}
\address{School of Mathematical Sciences, University of Science and Technology of China \newline \indent  Hefei, 230026, People's Republic of China}
\email{yqshi@ustc.edu.cn}

\author{Chunhui Wei$^\dagger$}
\address{School of Mathematics, Zhejiang University, Hangzhou 310058, People's Republic of China}
\email{chunhuiwei@zju.edu.cn}

\thanks{$^\dagger$C.W. is the corresponding author.}

\author{Bin Xu}
\address{School of Mathematical Sciences, University of Science and Technology of China \newline \indent  Hefei, 230026, People's Republic of China}
\email{bxu@ustc.edu.cn}
\begin{document}

\keywords{Toda systems; principal $\operatorname{SL}_2$-subgroups; flag manifolds; monodromy; spherical metrics}

\begin{abstract}
We study reduced solutions of Toda systems generated by spherical metrics for an arbitrary complex simple Lie algebra. The main result is an intrinsic centralizer parametrization: for a spherical metric with compact monodromy closure $G_v\subset\operatorname{PSU}(2)$, the reduced family is the positive Iwasawa slice cut out by the centralizer of the principal image of $G_v$. As an application, we transfer existing spherical-metric existence results to the diagonal Toda source tuple $(D,\ldots,D)$.
\end{abstract}
\maketitle
\tableofcontents

\section*{Introduction}

Toda systems associated with Cartan matrices arise in integrable systems, complex geometry, and representation theory. We begin by specifying the equation and the geometric objects used in the statements of the main results. Let $\mathfrak g$ be a complex simple Lie algebra of rank $n$, let $C=(a_{ij})$ be its Cartan matrix, and let $S$ be a Riemann surface. A $\mathfrak g$-Toda solution is an $n$-tuple of positive $(1,1)$-forms $\boldsymbol\varphi=(\varphi_1,\ldots,\varphi_n)$ satisfying
\[
\Ric(\boldsymbol\varphi)=2\boldsymbol\varphi C.
\]
If locally $\varphi_i=\frac{\sqrt{-1}}2e^{u_i}\,dz\wedge d\bar z$, this means
\[
\frac{\partial^2u_i}{\partial z\,\partial\bar z}
=-\sum_{j=1}^n a_{ji}e^{u_j},\qquad i=1,\ldots,n.
\]
For $\mathfrak g=\mathfrak{sl}_2(\CC)$ this is the equation for metrics of constant Gaussian curvature $4$. We call its solutions $\operatorname{SU}(2)$ Toda solutions, or spherical metrics.

Let $G$ be the simply connected complex group with Lie algebra $\mathfrak g$, let $K\subset G$ be a compact form, and fix a Borel subgroup $B_0\subset G$. The adjoint form of $G$ is
\[
G_\mathrm{ad}:=\mathrm{Ad}(G)\simeq G/Z(G)=\mathrm{Int}(\mathfrak{g}).
\]
Denote $K_{\mathrm{ad}}:=\operatorname{Ad}(K)$ and $B_{\mathrm{ad}}:=\operatorname{Ad}(B_0)$ for the corresponding images. The complete flag variety is $F=G/B_0=G_{\mathrm{ad}}/B_{\mathrm{ad}}$. Its standard holomorphic distribution is the $G$-invariant distribution generated at the base flag by the positive simple-root directions. A holomorphic curve tangent to it is called \emph{integral}, and is \emph{nondegenerate} when all its simple-root derivative components are nonzero.

For every fundamental representation $V_i$, there is a natural projective map $\pi_i:F\to\PP(V_i)$. We equip $\PP(V_i)$ with its $K$-invariant Fubini--Study form $\varphi_{\mathrm{FS},i}$. The generalized Pl\"ucker formulas say that a nondegenerate integral curve $f$ produces the Toda solution
\[
\bigl(f^*\pi_1^*\varphi_{\mathrm{FS},1},\ldots,
f^*\pi_n^*\varphi_{\mathrm{FS},n}\bigr)
\]
\cite{Positselskii1991,RazumovSaveliev1994,Nie2017}. On a non-simply connected surface, such a curve is an equivariant map $\widetilde f:\widetilde S\to F$ from the universal cover, together with its monodromy representation of $\pi_1(S)$ in $G_{\mathrm{ad}}$. We call it \emph{unitary} when its monodromy lies in $K_{\mathrm{ad}}$.

The Toda--flag correspondence rests on classical constructions. Positselskii established the local generalized Pl\"ucker formulas \cite{Positselskii1991}, while Razumov--Saveliev developed the flat-connection and reality framework for general simple groups \cite{RazumovSaveliev1994}; Nie related the same equations to integral curves of standard differential systems \cite{Nie2017}. Explicit holomorphic/projective-curve realizations for the classical types can be found in Gervais--Saveliev \cite{GervaisSaveliev1996} and Doliwa \cite{Doliwa1997}. Sijbrandij relates $\tau$-holomorphic flag curves to local open Toda frames \cite[Sections~4.2--4.3]{Sijbrandij2000}; his weak congruence theorem shows that their individual Toda invariants determine such curves up to a constant compact-group translation on an arbitrary Riemann surface \cite[Theorem~7.1]{Sijbrandij2000}. The nondegenerate holomorphic integral curves considered here belong to this $\tau$-holomorphic class.

For type $A_n$, the holomorphic-curve approach also yields classifications under polynomial energy growth: Eremenko treats the plane \mbox{\cite{Eremenko2007}}, and Zhao treats the punctured plane, relating the associated unitary curves to linear ordinary differential equations with Laurent-polynomial coefficients \mbox{\cite{Zhao2025}}.

Shi, Wei, and Xu \cite[Theorems~1.1 and~1.3]{ShiWeiXu2026} give the type-$A_n$ model: the rational normal (Veronese) curve composed with a spherical developing map produces the reduced solutions, and the parameter spaces are controlled by the spherical monodromy closure and its commutant. Building on this model, we study reduced solutions for every complex simple Lie algebra by combining the classical Toda--flag correspondence with the principal-$\operatorname{SL}_2$ centralizer calculation.

The earlier preprint of Shi--Wei--Xu \mbox{\cite{ShiWeiXu2022}} constructs the type-$A_n$ family from a locally univalent meromorphic function on a simply connected domain; the later work cited above treats its dependence on spherical monodromy.

The group-theoretic integration and zero-curvature viewpoint goes back to Leznov--Saveliev \cite{LeznovSaveliev1979}; see also their systematic treatment \cite[Chs.~3--4]{LeznovSaveliev1992}. The relation with primitive maps and Toda frames is developed further by Bolton--Pedit--Woodward \cite{BoltonPeditWoodward1995}.

The open Toda system uses a finite-type Cartan matrix; its untwisted affine extension adds the negative highest-root interaction. Cui, Wei, Yang, and Zhang \cite{CuiWeiYangZhang2023} study blow-up masses for affine $B_2^{(1)}$, while Cui, Nie, and Yang \cite{CuiNieYang2023} treat affine $A$ and $C^t$ types, relating local masses to affine Weyl groups. A direct link with open Toda theory occurs in bubbling: surviving components of rescaled affine $A_n^{(1)}$ solutions can converge to finite-energy open $A_\ell$ systems \cite[Proposition~2.2]{CuiNieYang2023}. Open Toda solutions thus provide local profiles in affine blow-up analysis. For open Toda systems, Cui, Gui, Jevnikar, Lin, and Yang \cite{CuiGuiJevnikarLinYang} describe the $D_n$ and $F_4$ Weyl-group mass sets used in blow-up analysis. Nie \mbox{\cite{Nie2025}} constructs explicit Toda solutions whose blow-up masses correspond to Weyl-group elements.

Finite-energy solutions on the plane with one singular source were classified in type $A_n$ by Lin--Wei--Ye \mbox{\cite{LinWeiYe2012}}, in types $B_n,C_n$ by Nie \mbox{\cite{Nie2016}}, and in type $G_2$ by Ao--Lin--Wei \mbox{\cite{AoLinWei2015}}. These classifications already include the reduced solutions satisfying that planar setup. In particular, Nie describes the parameter spaces as $AN_\Gamma$, with the allowed unipotent parameters determined by the source strengths and single-valuedness. Here we classify principal-curve families on arbitrary Riemann surfaces in terms of their full compact monodromy closure.

We first record the global, equivariant form of the classical Toda--flag correspondence. The local Toda connections glue to a flat $K_{\mathrm{ad}}$-bundle; the standard developing-map and holonomy construction then yields a nondegenerate integral curve on the universal cover with compact monodromy \cite[Ch.~II, \S\S3--4]{KobayashiNomizu1963}. The generalized Pl\"ucker formulas give the converse, and weak congruence gives uniqueness up to a constant $K_{\mathrm{ad}}$-translation.

\begin{introthm}
\label{introthm:global-correspondence}
Let $S$ be a Riemann surface. Every smooth $\mathfrak g$-Toda solution on $S$ is induced by a nondegenerate equivariant integral curve $\widetilde f:\widetilde S\to G_{\mathrm{ad}}/B_{\mathrm{ad}}$ whose monodromy lies in $K_{\mathrm{ad}}$. Two such curves induce the same Toda solution if and only if they differ by a constant left translation in $K_{\mathrm{ad}}$.
\end{introthm}

Theorem~\ref{introthm:global-correspondence} combines these established ingredients in the form needed for the reduced-solution classification. We include a proof that makes the coordinate cocycle, equivariance, and compact-monodromy quotient explicit.

The principal curve provides the analogue of the Veronese curve. A compact-compatible principal homomorphism
\[
\Phi_{\mathfrak g}:\operatorname{SL}_2(\CC)\longrightarrow G
\]
maps $\operatorname{SU}(2)$ into $K$ and defines an orbit curve $r_{\mathfrak g}:\PP^1\to F$, called the \emph{principal curve}. It also induces
\[
\overline\Phi_{\mathfrak g}:\operatorname{PSL}_2(\CC)\to G_{\mathrm{ad}},
\qquad
\sigma_{\mathfrak g}:
\operatorname{PSU}(2)\to K_{\mathrm{ad}}.
\]
An $\operatorname{SU}(2)$ Toda solution $\omega$ has a locally univalent unitary developing pair $v=(\widetilde v,\rho_v)$, where $\widetilde v:\widetilde S\to\PP^1$ is equivariant with respect to $\rho_v:\pi_1(S)\to\operatorname{PSU}(2)$ and $\widetilde v^*\varphi_{\mathrm{FS}}$ is the pullback of $\omega$ to $\widetilde S$. Throughout the Introduction and Sections~1--2, $v=(\widetilde v,\rho_v)$ denotes the developing pair, whereas $\widetilde v:\widetilde S\to\PP^1$ denotes its developing map.
\begin{introthm}
\label{introthm:principal-solution}
Let $\omega$ be an $\operatorname{SU}(2)$ Toda solution on a Riemann surface $S$, represented by a unitary developing pair $v=(\widetilde v,\rho_v)$. Then $r_{\mathfrak g}\circ\widetilde v$ has monodromy in $K_{\mathrm{ad}}$ and induces the $\mathfrak g$-Toda solution
\[
\boldsymbol\omega_{\mathfrak g}
=(m_1\omega,\ldots,m_n\omega),
\qquad
m_i=2\sum_{j=1}^n a^{ji},
\]
where $(a^{ij})$ is the inverse Cartan matrix.
\end{introthm}

We call a $\mathfrak g$-Toda solution \emph{reduced and generated by $\omega$} if one of its associated unitary flag curves is of the form $g\cdot r_{\mathfrak g}\circ\widetilde v$ for some $g\in G_{\mathrm{ad}}$. The classification problem is therefore to determine which transformations $g$ preserve unitary monodromy and when two of them induce the same Toda solution.

The global Toda--flag correspondence is then used to classify reduced solutions.

We next fix the developing pair $v$ and denote
\[
\Gamma_v=\rho_v(\pi_1(S)),
\qquad G_v=\overline{\Gamma_v}\subset\operatorname{PSU}(2).
\]
Let $G_{\mathrm{ad}}=K_{\mathrm{ad}}AN$ be an Iwasawa decomposition, with $A$ its real split torus and $N$ its unipotent factor. Let $\tau$ be the antiholomorphic involution of $G_{\mathrm{ad}}$ whose fixed group is $K_{\mathrm{ad}}$, denote $g^\dagger=\tau(g)^{-1}$, and denote $Z_{G_{\mathrm{ad}}}(H)$ for the centralizer of a subgroup $H\subset G_{\mathrm{ad}}$.

\begin{introthm}
\label{introthm:classification}
With the notation above, $G_v=\overline{\rho_v(\pi_1(S))}$ is a closed subgroup of $\operatorname{PSU}(2)\simeq\operatorname{SO}(3)$. The reduced $\mathfrak g$-Toda solutions generated by $\omega$ are parametrized bijectively by
{
\[
\Delta_v^{\mathfrak g}
=\left\{\delta\in AN:
\delta^\dagger\delta\in
Z_{G_{\mathrm{ad}}}\bigl(\sigma_{\mathfrak g}(G_v)\bigr)
\right\}.
\]}
The parameter $\delta$ corresponds to the curve $\delta\cdot r_{\mathfrak g}\circ\widetilde v$. If
$\mathcal E_{\mathfrak g}=\{d_1,\ldots,d_n\}$ is the exponent multiset of $\mathfrak g$, with multiplicities, then
{
\[
\dim_{\RR}\Delta_v^{\mathfrak g}
=\sum_{d\in\mathcal E_{\mathfrak g}}
\dim_{\CC}\left(\operatorname{Sym}^{2d}(\CC^2)\right)^{G_v}.
\]}
The explicit cases for all closed subgroups of $\operatorname{PSU}(2)$ are collected in Section~3.
\end{introthm}

Here the exponents are characterized by the decomposition of the adjoint representation under the principal $\mathfrak{sl}_2$ into irreducibles of dimensions $2d+1$; Kostant's theorem gives the displayed sum \cite{Kostant1959}.

The dimension formula above yields the following comparisons. If $G_v=\operatorname{PSU}(2)$, only the principal solution occurs (Theorem~\ref{thm:classification-reduced-solutions}(i)); a maximal torus yields a family of dimension $\operatorname{rank}\mathfrak g$, and the torus normalizer has a one-dimensional fixed space precisely in the even-exponent summands (Proposition~\ref{prop:fixed-space-dictionary}). The torus-normalizer dimension is $\lfloor n/2\rfloor$ in type $A_n$, zero in types $B_n,C_n,E_7,E_8,F_4,G_2$, equal to $2$ in type $E_6$, and equal to $1$ in type $D_n$ precisely when $n$ is odd (Corollary~\ref{cor:all-types-dimension-table}).

Section~3 evaluates the finite-group terms and concludes with a table covering all simple types, including $A_n$. The equality of the exponent lists for $B_n$ and $C_n$ shows that their dimensions agree for every monodromy closure. For polyhedral monodromy, the dimensions in the exceptional types are:
\[
\begin{array}{c|ccccc}
G_v&E_6&E_7&E_8&F_4&G_2\\ \hline
\mathsf T&4&8&16&2&0\\
\mathsf O&2&3&6&0&0\\
\mathsf I&0&0&0&0&0.
\end{array}
\]
In particular, icosahedral monodromy is rigid in every exceptional type; octahedral monodromy is also rigid in types $F_4$ and $G_2$. Applying the character calculation to the classical families yields explicit quasipolynomials in the rank: type $A_n$ uses the full exponent list $1,\ldots,n$, types $B_n,C_n$ use $1,3,\ldots,2n-1$, and type $D_n$ has the additional exponent $n-1$.

The construction also transfers existing existence results for spherical metrics with conical points to singular Toda systems. This generalizes the type-$A$ application in \cite{ShiWeiXu2026}, but the transfer is not presented as a new PDE existence theory.

\begin{introthm}
\label{introthm:cone-application}
{
Let $X$ be a compact Riemann surface carrying a cone spherical metric $\omega$ with divisor $D=\sum_{\nu=1}^m\gamma_\nu[P_\nu]$, where $\gamma_\nu>-1$. For every complex simple Lie algebra $\mathfrak g$ of rank $n$, the $\mathfrak g$-Toda system on $X$ with cone-divisor tuple $(D,\ldots,D)$ admits a reduced family parametrized by $\Delta_v^{\mathfrak g}$ and containing
\[
(m_1\omega,\ldots,m_n\omega),
\qquad
m_i=2\sum_{j=1}^n a^{ji}.
\]
In particular, if $g>0$ and $\beta_1,\ldots,\beta_m>0$ satisfy $\sum_{\nu=1}^m\beta_\nu>2g-2+m$, such a reduced family exists on some genus-$g$ compact Riemann surface with distinct points $P_1,\ldots,P_m$ and $D=\sum_{\nu=1}^m(\beta_\nu-1)[P_\nu]$.}
\end{introthm}

Here the cone divisor $D=\sum_\nu\gamma_\nu[P_\nu]$ records the local behavior of the metric: in a coordinate centered at $P_\nu$, its density is $|z|^{2\gamma_\nu}$ times a positive continuous function. Theorem~D concerns the diagonal cone-divisor tuple $(D,\ldots,D)$, with the same marked points and orders in every component; it does not address arbitrary source matrices. The resulting reduced family may be zero-dimensional, and its dimension is controlled by the monodromy closure through Theorem~C. The numbers $\beta_\nu=\gamma_\nu+1$ are the corresponding cone-angle parameters.

The positive-genus existence statement combines the transfer theorem with the spherical-metric existence theorem of Mondello and Panov \cite[Theorem~A]{MondelloPanov2019}, after rescaling their curvature-one metric by the factor $1/4$ to match our curvature-four normalization. When the spherical metric is the pullback of $\varphi_{\mathrm{FS}}$ by a holomorphic map to $\PP^1$, its monodromy is trivial and the parameter space is the entire group $AN$, of real dimension $\dim_{\CC}\mathfrak g$.

The paper is organized as follows. Section~\ref{sec:basic} establishes the global correspondence between Toda solutions and nondegenerate unitary integral curves. Section~\ref{sec:reduced} introduces reduced solutions, computes the Toda forms of the principal curve, proves the monodromy and rigidity lemmas, and derives the general parameterization and dimension formula. Section~\ref{sec:classification} records the resulting type-by-type dimension formulas, the finite-monodromy character calculation, and a comparison table for all simple types, including $A_n$. Section~\ref{sec:application} applies the master classification to conical Toda systems on compact Riemann surfaces, including positive-genus examples. Section~\ref{sec:counterexample} considers a converse question and provides a counterexample.

\section{Basic Correspondence}\label{sec:basic}
We establish the global correspondence between solutions of the Toda system and nondegenerate integral curves with compact monodromy. Its local form is standard; see Positselskii \cite{Positselskii1991}, Razumov--Saveliev \cite{RazumovSaveliev1994}, and Nie \cite{Nie2017}; the related congruence principle is discussed by Sijbrandij \cite[Theorem~7.1]{Sijbrandij2000}.

\subsection{Review of Lie Theory}\label{sec:Lie theory}
Let $G$ be a complex simply connected simple Lie group of rank $n$, $K$ a compact form, and $\mathfrak g,\mathfrak k$ their Lie algebras. Fix a triangular decomposition
\[
\mathfrak g=\mathfrak h\oplus\mathfrak n^+\oplus\mathfrak n^-,
\]
and let $B_0$ be the Borel subgroup with Lie algebra $\mathfrak b_0=\mathfrak h\oplus\mathfrak n^-$. We also denote $\mathfrak b=\mathfrak b_0$; this negative-Borel convention is used throughout the paper. $B_{\mathrm{ad}}$ denotes the image of $B_0$ in the adjoint group $G_{\mathrm{ad}}=G/Z(G)$. Let $F$ be the variety of Borel subgroups of $G$, with base point $B_0$. A character $\chi:B_0\to\mathbb C^*$ determines the $G$-equivariant line bundle
\[
\mathcal L(\chi)=G\times_{B_0}\mathbb C,
\qquad (gb,v)\sim(g,\chi(b)v),\quad b\in B_0,
\]
and every $G$-equivariant line bundle on $F$ arises in this way.

Define the $G$-action on $\Gamma(F,\mathcal{L}(\chi))$ by
\[
(g\cdot\sigma)(B)=g\cdot\sigma(g^{-1}Bg).
\]
\begin{thm}[Borel--Weil Theorem {\cite[\S 23]{FultonHarris1991}}]
Let $\chi:B_0\to\mathbb C^*$ be a character whose differential $\lambda=d\chi\in\mathfrak h^*$ is dominant, and let $V_\lambda$ denote the irreducible representation of highest weight $\lambda$. For our negative-Borel and associated-character conventions,
\[
\Gamma(F,\mathcal L(\chi))\simeq V_\lambda,
\qquad
\Gamma(F,\mathcal L(\chi))^*\simeq V_\lambda^*\simeq V_{-w_0\lambda},
\]
where $w_0$ is the longest Weyl group element.
\end{thm}
Let $\omega_i:B_0\to\mathbb{C}^*$, $i=1,\ldots,n$, be the fundamental characters, and let $\varpi_i=d\omega_i$ be their differentials, viewed as the fundamental weights. Denote $\mathcal L_i=\mathcal L(\omega_i)$ and
\[
V_i:=\Gamma(F,\mathcal L_i)^*\simeq V_{-w_0\varpi_i}.
\]
Thus the $V_i$ are the opposition-indexed fundamental representations. Although $V_i$ need not be a linear representation of $G_{\mathrm{ad}}$, the center of $G$ acts on the irreducible module $V_i$ by scalars. Hence the projective action on $\mathbb P(V_i)$ descends to $G_{\mathrm{ad}}$.

Let $V$ be a finite-dimensional irreducible representation of $G$. A $K$-invariant Hermitian form on $V$ is unique up to multiplication by a positive constant. Since such a rescaling does not change the induced projective metric, this determines a unique $K$-invariant Fubini--Study form $\varphi_{\mathrm{FS}}$ on $\mathbb{P}(V)$; see \cite[\S9.3, p.~130]{FultonHarris1991} and Schur's Lemma.

Let $\pi_i:F\to\mathbb{P}(V_i)$ be the projective morphism defined by evaluation. For each $B\in F$, the evaluation quotient
    \[
    \operatorname{ev}_B:\Gamma(F,\mathcal L_i)
    \longrightarrow \mathcal L_i|_B
    \]
determines a line in $\Gamma(F,\mathcal L_i)^*=V_i$ after dualizing. We denote $\pi_i(B)=[\operatorname{ev}_B]$; changing a trivialization of the one-dimensional fiber $\mathcal L_i|_B$ only rescales the representative.
    \subsection{Local Pl\"ucker Formulas}
The local Pl\"ucker formulas for semisimple groups were proved by Positsel'skii \cite{Positselskii1991}; related differential-geometric formulations and extensions appear in \cite{RazumovSaveliev1994,BoltonWoodward2003}. Let $f:S\to F$ be a holomorphic curve and $\varphi_i=f^*\pi_i^*\varphi_{\mathrm{FS}}$. Work near a point where all these forms are positive, and denote $\Ric(\boldsymbol\varphi)=(\Ric(\varphi_1),\ldots,\Ric(\varphi_n))$.

The tangent space $T_BF$ is isomorphic to $\mathfrak{g}/\mathfrak{b}$, where $\mathfrak{b}$ is the Lie algebra of the group $B$. Define a natural $n$-dimensional distribution $\mathcal{N}$ on the flag variety $F$ by
\[
\mathcal{N}(B)=\{x\in\mathfrak{g}|[x,[\mathfrak{b},\mathfrak{b}]]\subset\mathfrak{b}\}/\mathfrak{b}.
\]
For each simple root, let $P_i(B)$ be the corresponding minimal parabolic subgroup containing $B$, with Lie algebra $\mathfrak p_i(B)$, and denote
\[
\mathcal N_i(B):=\mathfrak p_i(B)/\mathfrak b.
\]
Then $\mathcal N=\bigoplus_{i=1}^n\mathcal N_i$; denote the associated projections by $pr_i:\mathcal N\to\mathcal N_i$. We will show the details of this decomposition in the proof of Theorem \ref{conj:Givental}.
\begin{defn}
    An integral curve $f:S\to F$ of the distribution $\mathcal{N}$ is called \emph{nondegenerate} if, for every simple-root summand $\mathcal N_i$,
\[
pr_i\circ df:TS\longrightarrow f^*\mathcal N_i
\]
is nowhere vanishing.
\end{defn}
In the following text, when we refer to an integral curve $f:S\to F$, we assume by default that it is tangent to the distribution $\mathcal{N}$.

Only nondegenerate integral curves induce an $n$-tuple of positive K\"ahler forms. The following local Pl\"ucker formula was conjectured by Givental \cite{Givental1989} and proved by Positsel'skii \cite{Positselskii1991}.
    \begin{thm}[Givental's conjecture]\label{conj:Givental}
If $f:S\to F$ is a nondegenerate integral curve on the flag manifold $F$, then
    \[
    \Ric(\boldsymbol\varphi)=2\boldsymbol\varphi C.
    \]
    \end{thm}
    \begin{proof}[Review of Positsel'skii's Proof]\quad\\
At $B_0$, the condition $[x,\mathfrak n^-]\subset\mathfrak b_0$ selects exactly the simple positive-root directions. Indeed, if $\beta$ is a nonsimple positive root, then $\beta-\alpha_k$ is positive for some simple root $\alpha_k$, and $[\mathfrak g_\beta,\mathfrak g_{-\alpha_k}]\not\subset\mathfrak b_0$; whereas $[\mathfrak g_{\alpha_i},\mathfrak n^-]\subset\mathfrak b_0$. Thus
\[
\mathcal N_i(B_0)=\mathfrak g_{\alpha_i}\bmod\mathfrak b_0,
\qquad
\mathcal N=\bigoplus_i\mathcal N_i,
\]
and the $B_0$-character of the fiber $\mathcal N_i(B_0)$ is $\alpha_i$.

Let $[v_i]=\pi_i(B_0)$. The evaluation construction makes this a lowest-weight line of weight $-\varpi_i$ in $V_i$. Choose nonzero simple-root vectors $e_j\in\mathfrak g_{\alpha_j}$. The corresponding $\mathfrak{sl}_2$-representations give $e_jv_i=0$ for $j\ne i$ and $e_iv_i\ne0$. Since the differential of the projective orbit map sends $X$ to $Xv_i\bmod\mathbb C v_i$, equivariance yields
\[
d\pi_i(\mathcal N_j)=0\quad(j\ne i),
\qquad
d\pi_i|_{\mathcal N_i}\ne0
\quad\text{everywhere on }F.
\]

We use the curvature convention
\[
\theta(\mathcal L,h)
=-\sqrt{-1}\,\partial\bar\partial\log\|s\|_h^2
\]
for a local nonvanishing holomorphic section $s$. Evaluation identifies $\mathcal L_i\simeq\pi_i^*\mathcal O_{\mathbb P(V_i)}(1)$. Put $h_i=\pi_i^*h_{\mathrm{FS}}$, where $\theta(\mathcal O(1),h_{\mathrm{FS}})=2\varphi_{\mathrm{FS}}$. Then
\[
\theta(\mathcal L_i,h_i)=2\pi_i^*\varphi_{\mathrm{FS}}.
\]
The root--weight identity $\alpha_i=\sum_j a_{ji}\varpi_j$ and the fiber characters give
\[
\mathcal N_i\simeq\bigotimes_j\mathcal L_j^{\otimes a_{ji}}.
\]
Equip this bundle with the tensor-product metric $h_{\mathcal N_i}=\bigotimes_j h_j^{a_{ji}}$, using dual metrics for negative powers.

Nondegeneracy makes $pr_i\circ df:TS\to f^*\mathcal N_i$ a holomorphic line-bundle isomorphism. The tensor-product metric on $\mathcal N_i$ and the metric induced through $d\pi_i$ are $K$-invariant. Since $K$ acts transitively on $F$, their ratio is a positive constant. By the differential identities above, the latter pulls back under $pr_i\circ df$ to the tangent metric defined by $\varphi_i$. Constant rescaling leaves curvature unchanged; hence additivity of curvature under tensor products gives
\[
\Ric(\varphi_i)
=f^*\theta(\mathcal N_i,h_{\mathcal N_i})
=2\sum_j a_{ji}\varphi_j.
\]
This is the asserted formula in each component.
    \end{proof}
\subsection{From Solutions to Integral Curves}
We now realize the Toda/integral-curve correspondence through the standard zero-curvature construction, as in Razumov--Saveliev \cite{RazumovSaveliev1994} and Nie \cite{Nie2017}; see also the group-theoretic treatment of Leznov--Saveliev \cite[Chs.~3--4]{LeznovSaveliev1992}. Assume $u=(u_1,\ldots,u_n)$ is a solution of the Toda system of type $\mathfrak g$ on a simply connected domain $U\subset\CC$, and use the triangular decomposition fixed in Section~\ref{sec:Lie theory}. Let $\tau$ be the antiholomorphic involution whose fixed-point set is the chosen compact real form. Choose the Chevalley generators compatibly with this compact form, so that
    \[
    \tau(e_i)=-f_i,\qquad \tau(f_i)=-e_i,
    \qquad \tau(h_i)=-h_i.
    \]
Thus $(e_i,h_i,f_i)$ is an $\mathfrak{sl}_2$-triple satisfying
    \[
    [h_i,e_i]=2e_i,[h_i,f_i]=-2f_i,[e_i,f_i]=h_i.
    \]
Define a $\mathfrak{g}$-valued $1$-form
    \[
    \Theta=\frac{1}{2}\sum_{i,j}a^{ij}h_j
    (\partial_z u_i\,dz-\partial_{\bar z}u_i\,d\bar z)
    +\sum_i e^{u_i/2}(e_i\,dz-f_i\,d\bar z),
    \]
where $(a^{ij})=C^{-1}$. For every real tangent vector $X$, the compact compatibility above and the reality of the functions $u_i$ give $\tau(\Theta(X))=\Theta(X)$. Hence $\Theta$ is valued in the compact real form as a real one-form.
\begin{lem}
$\Theta$ satisfies the Maurer--Cartan equation:
    \[
        d\Theta+\frac{1}{2}[\Theta,\Theta]=0.
    \]
\end{lem}
\begin{proof}
Define $\Theta=\Theta_z\,dz+\Theta_{\bar z}\,d\bar z$. Using $[h_j,e_i]=a_{ji}e_i$, $[h_j,f_i]=-a_{ji}f_i$, and $[e_i,f_j]=\delta_{ij}h_i$, a direct calculation gives
\[
\partial_z\Theta_{\bar z}-\partial_{\bar z}\Theta_z
+[\Theta_z,\Theta_{\bar z}]
=-\sum_j\left(\sum_i a^{ij}\partial_z\partial_{\bar z}u_i
+e^{u_j}\right)h_j.
\]
The Toda equation implies $\sum_i a^{ij}\partial_z\partial_{\bar z}u_i=-e^{u_j}$, so the right-hand side vanishes.
\end{proof}
Fix $z_0\in U$ and an initial value $\phi_U(z_0)\in K$. There exists a unique $\phi_U:U\to G$ such that $\Theta|_U=\phi_U^{-1}d\phi_U$ with this initial value. Since $\Theta$ is compact-real-valued, uniqueness for this differential equation gives $\phi_U(U)\subset K$. Define $f_U:U\to F$ by composing $\phi_U$ with $pr:G\to F$, $g\mapsto gB_0g^{-1}$. For the remainder of the local argument we simply denote $\phi$ and $f$ for $\phi_U$ and $f_U$.
\begin{lem}
$f$ is a holomorphic integral curve of the distribution $\mathcal{N}$.
\end{lem}
\begin{proof}
Under the natural identification
\[
    T_{f(z)}F\simeq\operatorname{Ad}_{\phi(z)}(\mathfrak g/\mathfrak b),
\]
the differential is
\[
df(X)=\operatorname{Ad}_{\phi(z)}\left(\Theta(X)\mod \mathfrak b\right).
\]
Because $\mathfrak b=\mathfrak h\oplus\mathfrak n^-$,
\[
\Theta(\partial_{\bar z})\in\mathfrak b,
\qquad
\Theta(\partial_z)\equiv\sum_i e^{u_i/2}e_i\pmod{\mathfrak b}.
\]
Thus $df(\partial_{\bar z})=0$, while $df(\partial_z)$ lies in the $G$-translate of $\bigoplus_i\mathfrak g_{\alpha_i}\bmod\mathfrak b=\mathcal N(B_0)$. Therefore $f$ is holomorphic and integral. Since every coefficient $e^{u_i/2}$ is positive, it is also nondegenerate.
\end{proof}
\begin{lem}
    \[
    f^*\pi_i^*\varphi_{\mathrm{FS}}
    =\frac{\sqrt{-1}}2e^{u_i} dz \wedge d\bar{z}.
    \]
\end{lem}
\begin{proof}
Let
\[
\Theta_{\mathfrak h}=\frac12\sum_{r,s}a^{rs}h_s
\left(
\partial_z u_r\,dz-\partial_{\bar z}u_r\,d\bar z
\right)
\]
be the Cartan component of $\Theta$.
Denote $T^c=K\cap B_0$, so that $F\simeq K/T^c$ and $\mathcal L_i\simeq K\times_{T^c}\mathbb C_{\omega_i}$. Applying $d\omega_i$ to the $\operatorname{Lie}(T^c)$-component of the Maurer--Cartan form gives a unitary connection on this homogeneous line bundle. Its $(0,1)$-part agrees with the given holomorphic structure: for a local $K$-lift of a holomorphic map, the $(0,1)$-component lies in $\mathfrak b_0=\mathfrak h\oplus\mathfrak n^-$, and $d\omega_i$ vanishes on $\mathfrak n^-$. Hence this is the Chern connection. Since any two $K$-invariant Hermitian metrics on $\mathcal L_i$ differ by a positive constant, it is also the Chern connection of $h_i$.

In the unitary gauge determined by $\phi$, the connection form on $f^*\mathcal L_i$ is
\[
 A_i=d\omega_i(\Theta_{\mathfrak h})
 =\frac12\sum_r a^{ri}
 (\partial_z u_r\,dz-\partial_{\bar z}u_r\,d\bar z).
\]
With the curvature convention of Step~4, $\theta(f^*\mathcal L_i,f^*h_i)=\sqrt{-1}\,dA_i$.
The Fubini--Study normalization and the Toda equation now give
\[
\begin{aligned}
f^*\pi_i^*\varphi_{\mathrm{FS}}
&=\frac{\sqrt{-1}}2\,dA_i
=-\frac{\sqrt{-1}}2\sum_r a^{ri}
\partial_z\partial_{\bar z}u_r\,dz\wedge d\bar z\\
&=\frac{\sqrt{-1}}2\sum_{r,s}a^{ri}a_{sr}e^{u_s}\,dz\wedge d\bar z \\
&=\frac{\sqrt{-1}}2e^{u_i}dz\wedge d\bar z.
\end{aligned}
\]
\end{proof}
Hence every solution of the Toda system on $U$ arises locally from a nondegenerate integral curve. Define
\begin{equation}\label{eq:rho-dual}
\rho^\vee:=
\sum_{j=1}^n\left(\sum_{i=1}^n a^{ij}\right)h_j.
\end{equation}
Then $\alpha_i(\rho^\vee)=1$ for every simple root $\alpha_i$.
\begin{thm}\label{thm:toda-to-unitary-curve}
Let $(\varphi_1,\ldots,\varphi_n)$ be a smooth solution of the Toda system of type $\mathfrak g$ on a Riemann surface $S$. Then there is a multi-valued integral curve represented by an equivariant pair
     \[
     \widetilde f:\widetilde S\to F,\qquad
     \rho_f:\pi_1(S)\to K_{\mathrm{ad}},
     \]
where $p:\widetilde S\to S$ is the universal covering map, such that $\widetilde f(\gamma x)=\rho_f(\gamma)\widetilde f(x)$ and $p^*\varphi_i=\widetilde f^*\pi_i^*\varphi_{\mathrm{FS}}$ on $\widetilde S$. Furthermore, for every $i=1,\ldots,n$, the map
        \[
        pr_i\circ d\widetilde f:T\widetilde S
        \longrightarrow\widetilde f^*\mathcal{N}_i
        \]
is nowhere vanishing.
\end{thm}
\begin{proof}
First construct the connection on $S$. If $z,w$ are overlapping holomorphic coordinates and $\lambda=\frac{dz}{dw}=|\lambda|e^{\sqrt{-1}\vartheta}$, then $u_i^{(w)}=u_i^{(z)}+\log|\lambda|^2$. Because $dw=\lambda^{-1}dz$ and $e^{u_i^{(w)}/2}=|\lambda|e^{u_i^{(z)}/2}$, the simple-root terms transform according to
\[
\begin{aligned}
e^{u_i^{(w)}/2}e_i\,dw
&=e^{-\sqrt{-1}\vartheta}e^{u_i^{(z)}/2}e_i\,dz,\\
e^{u_i^{(w)}/2}f_i\,d\bar w
&=e^{\sqrt{-1}\vartheta}e^{u_i^{(z)}/2}f_i\,d\bar z.
\end{aligned}
\]
The change in the Cartan part is
\[
\frac12\rho^\vee\left(
\partial_w\log|\lambda|^2\,dw
-\partial_{\bar w}\log|\lambda|^2\,d\bar w
\right)=\sqrt{-1}\rho^\vee\,d\vartheta
\]
where we used that $\lambda$ is holomorphic and nowhere zero on the coordinate overlap. Denote
\[
t_{zw}=\exp(\sqrt{-1}\vartheta\,\rho^\vee)\in K_{\mathrm{ad}}.
\]
Since $\alpha_i(\rho^\vee)=1$, one has
\[
\operatorname{Ad}_{t_{zw}^{-1}}e_i
=e^{-\sqrt{-1}\vartheta}e_i,
\qquad
\operatorname{Ad}_{t_{zw}^{-1}}f_i
=e^{\sqrt{-1}\vartheta}f_i,
\qquad
t_{zw}^{-1}dt_{zw}=\sqrt{-1}\rho^\vee\,d\vartheta.
\]
The preceding calculations therefore give the gauge law
\[
\Theta^{(w)}
=t_{zw}^{-1}\Theta^{(z)}t_{zw}+t_{zw}^{-1}dt_{zw}.
\]
Replacing $\vartheta$ by $\vartheta+2\pi$ multiplies $t_{zw}$ by $\exp_{G_{\mathrm{ad}}}(2\pi\sqrt{-1}\rho^\vee)$. Since every root is an integral combination of simple roots and $\alpha_i(\rho^\vee)=1$, the element $\exp_G(2\pi\sqrt{-1}\rho^\vee)$ acts trivially in the adjoint representation and therefore belongs to $Z(G)$. Consequently, $\exp_{G_{\mathrm{ad}}}(2\pi\sqrt{-1}\rho^\vee)=e$, and $t_{zw}$ is independent of the chosen argument of $\lambda$. On a triple overlap, the identities
\[
\frac{dz}{d\zeta}=\frac{dz}{dw}\frac{dw}{d\zeta}
\]
and additivity of the arguments modulo $2\pi$ give $t_{z\zeta}=t_{zw}t_{w\zeta}$ in $G_{\mathrm{ad}}$. Hence these transition functions satisfy the cocycle condition and define a principal $K_{\mathrm{ad}}$-bundle $P\to S$, and the local forms $\Theta^{(z)}$ define on it a global flat connection. Compact compatibility makes $\Theta^{(z)}$ $\mathfrak k_{\mathrm{ad}}$-valued as a real one-form, so its holonomy lies in $K_{\mathrm{ad}}$. This is the standard developing-map/holonomy construction for a flat principal connection \cite[Ch.~II, \S\S3--4]{KobayashiNomizu1963}.

Pull $P$ and its connection back by $p:\widetilde S\to S$. Since $\widetilde S$ is simply connected, parallel transport trivializes the pullback. On a coordinate chart $U\subset\widetilde S$, the preceding local construction gives a map $f_U:U\to F$ satisfying
\[
p^*\varphi_i|_U=f_U^*\pi_i^*\varphi_{\mathrm{FS}}.
\]
The local maps glue in the parallel trivialization to a single holomorphic integral curve $\widetilde f:\widetilde S\to F$ together with the holonomy representation $\rho_f:\pi_1(S)\to K_{\mathrm{ad}}$ of the connection on $P$ such that
\[
\widetilde f(\gamma\cdot x)
=
\rho_f(\gamma)\widetilde f(x).
\]
The local differential formula shows that every simple-root coefficient is $e^{u_i/2}>0$. Hence each $pr_i\circ d\widetilde f$ is nowhere vanishing. This equivariant pair is, by definition, the required multi-valued curve.
\end{proof}

The following proposition is the weak congruence principle for the nondegenerate holomorphic integral curves considered here; see Sijbrandij \cite[Section~4.3 and Theorem~7.1]{Sijbrandij2000}. We give a direct proof in our adjoint setting, expressing the invariants through the fundamental Pl\"ucker forms and making the frame normalization and coordinate compatibility explicit.

\begin{prop}
\label{prop:unitary-curve-rigidity}
Let $\widetilde f_1,\widetilde f_2:\widetilde S\to G_{\mathrm{ad}}/B_{\mathrm{ad}}$ be nondegenerate integral curves whose monodromy lies in the fixed $K_{\mathrm{ad}}$. If
\[
\widetilde f_1^*\pi_i^*\varphi_{\mathrm{FS}}
=\widetilde f_2^*\pi_i^*\varphi_{\mathrm{FS}}
\qquad\text{for every }i,
\]
then there is a constant $k\in K_{\mathrm{ad}}$ such that
\[
\widetilde f_2=k\cdot\widetilde f_1.
\]
\end{prop}

\begin{proof}
Work first on a simply connected coordinate disc $D$ and denote the common forms as
\[
\frac{\sqrt{-1}}2e^{u_i}dz\wedge d\bar z.
\]
For either curve, choose a local lift $\psi:D\to K_{\mathrm{ad}}$ through $K_{\mathrm{ad}}\to G_{\mathrm{ad}}/B_{\mathrm{ad}}$ and denote its Maurer--Cartan form as $\beta=\psi^{-1}d\psi=\beta_z\,dz+\beta_{\bar z}\,d\bar z$. Holomorphicity gives $\beta_{\bar z}\in\mathfrak b$. Since $\beta$ is $\mathfrak k_{\mathrm{ad}}$-valued as a real one-form, compact reality gives $\tau(\beta_z)=\beta_{\bar z}$, and hence $\beta_z\in\tau(\mathfrak b)=\mathfrak h\oplus\mathfrak n^+$. The integral-curve condition removes all positive-root components except the simple-root components, while nondegeneracy makes each of their coefficients nonzero. Thus
\[
\beta_z=H_z+\sum_i c_i e_i,\qquad
\beta_{\bar z}=\tau(H_z)-\sum_i\overline{c_i}f_i,
\]
where $H_z\in\mathfrak h$ and $c_i\neq0$. With the Chevalley vectors and homogeneous metrics fixed above, the relevant normalization is
\[
\theta(\mathcal O(1),h_{\mathrm{FS}})=2\varphi_{\mathrm{FS}}.
\]
We record why this normalization introduces no further constant in the differential. Under the dual action on $V_i=\Gamma(F,\mathcal L_i)^*$, the evaluation line at $B_0$ has $B_0$-character $\omega_i^{-1}$ and hence infinitesimal weight $-\varpi_i$. Thus a unit vector $v_i$ in this line satisfies
\[
f_jv_i=0,\qquad h_jv_i=-\delta_{ij}v_i.
\]
In particular, in the simple-root $\mathfrak{sl}_2$-subalgebra,
\[
f_i v_i=0,\qquad h_i v_i=-v_i.
\]
For the compact-compatible Hermitian form one has $e_i^*=f_i$, and hence
\[
\lVert e_i v_i\rVert^2
=\langle v_i,f_i e_i v_i\rangle
=\langle v_i,-h_i v_i\rangle
=\lVert v_i\rVert^2=1.
\]
Moreover, $e_i v_i\perp v_i$, since the two vectors lie in distinct weight spaces for the compact torus and distinct compact-torus weight spaces are orthogonal.
Thus $d\pi_i:\mathcal N_i\to T\mathbb P(V_i)$ sends the normalized root vector $e_i$ to a unit tangent vector for the chosen Fubini--Study metric. Consequently, the pullback equality gives $|c_i|^2=e^{u_i}$. Denote $c_i=q_i e^{\sqrt{-1}\vartheta_i}$, where $q_i=e^{u_i/2}>0$.

Let
\[
T_{\mathrm{ad}}^c:=K_{\mathrm{ad}}\cap B_{\mathrm{ad}}
\]
be the compact maximal torus determined by the chosen Borel. The character map
\[
T_{\mathrm{ad}}^c\longrightarrow(S^1)^n,
\qquad t\longmapsto(\alpha_1(t),\ldots,\alpha_n(t)),
\]
is an isomorphism of compact tori because the simple roots form a basis of the character lattice of the adjoint torus. It therefore supplies the unique gauge satisfying $\alpha_i(t)=e^{\sqrt{-1}\vartheta_i}$ for all $i$, and this gauge makes all the simple-root coefficients equal to $q_i=e^{u_i/2}$.

This gauge is unique in the adjoint group. Indeed, if two compact-torus gauges have the same positive simple-root coefficients, their ratio $t$ satisfies $\alpha_i(t)=1$ for every simple root. It therefore lies in the center, which is trivial in $G_{\mathrm{ad}}$.

In the resulting normalized adapted frame, compact reality and the integral-curve condition give
\[
\beta_z=H_z+\sum_i e^{u_i/2}e_i,
\qquad
\beta_{\bar z}=H_{\bar z}-\sum_i e^{u_i/2}f_i,
\]
where $H_z,H_{\bar z}\in\mathfrak h$. Taking the $e_i$- and $f_i$-components of $d\beta+\frac12[\beta\wedge\beta]=0$ yields
\[
\alpha_i(H_z)=\frac12\partial_z u_i,
\qquad
\alpha_i(H_{\bar z})=-\frac12\partial_{\bar z}u_i.
\]
Since the simple roots form a basis of $\mathfrak h^*$, inversion of the Cartan matrix determines the Cartan terms uniquely. Under the conventions fixed in Section~\ref{sec:Lie theory},
\[
H_z=\frac12\sum_{i,j}a^{ij}(\partial_z u_i)h_j,
\qquad
H_{\bar z}=-\frac12\sum_{i,j}a^{ij}
(\partial_{\bar z}u_i)h_j.
\]
Consequently both normalized frames have exactly the same Maurer--Cartan form
\[
\frac12\sum_{i,j}a^{ij}h_j
(\partial_z u_i\,dz-\partial_{\bar z}u_i\,d\bar z)
+\sum_i e^{u_i/2}(e_i\,dz-f_i\,d\bar z).
\]
They therefore differ by a constant left factor; since the two frames take values in $K_{\mathrm{ad}}$, that factor belongs to $K_{\mathrm{ad}}$.

It remains to check that these constants patch. On an overlap of coordinates $z,w$, both normalized frames transform by the same torus gauge $t_{zw}=\exp(\sqrt{-1}\vartheta\rho^\vee)$. Thus, if $\psi_2^{(z)}=k_z\psi_1^{(z)}$, then
\[
\psi_2^{(w)}=\psi_2^{(z)}t_{zw}
=k_z\psi_1^{(z)}t_{zw}=k_z\psi_1^{(w)},
\]
and consequently $k_w=k_z$. Since $\widetilde S$ is connected, the local constants patch to one element $k\in K_{\mathrm{ad}}$ on all of $\widetilde S$.
\end{proof}

\begin{defn}
Denote by $\mathrm{Int}_{G_{\mathrm{ad}}}(S)$ the space of nondegenerate integral maps $\widetilde f:\widetilde S\to F$ for which there exists an equivariance representation $\rho_f:\pi_1(S)\to G_{\mathrm{ad}}$ (the auxiliary representation is not part of the datum in this space). Let $\mathrm{Int}_{G_{\mathrm{ad}}}(S)_{K_{\mathrm{ad}}}$ be the subspace of maps admitting such a representation with image in the fixed compact group $K_{\mathrm{ad}}$, and let $\mathrm{Toda}_{\mathfrak g}(S)$ be the space of solutions of the $\mathfrak g$-Toda system on $S$.
\end{defn}
Pullback of the fundamental Fubini--Study forms defines the map
\[
\mathrm{Int}_{G_\mathrm{ad}}(S)_{K_\mathrm{ad}}
\longrightarrow\mathrm{Toda}_\mathfrak{g}(S),
\qquad
\widetilde f\longmapsto(\varphi_i)_{i=1}^n,
\qquad
p^*\varphi_i=\widetilde f^*\pi_i^*\varphi_{\mathrm{FS}}.
\]
The map is invariant under the left action of $K_{\mathrm{ad}}$.
\begin{cor}
The pullback map induces a bijection
\[
K_{\mathrm{ad}}\backslash
\mathrm{Int}_{G_{\mathrm{ad}}}(S)_{K_{\mathrm{ad}}}
\xrightarrow{\ \sim\ }
\mathrm{Toda}_{\mathfrak g}(S).
\]
\end{cor}
\begin{proof}
Surjectivity is Theorem~\ref{thm:toda-to-unitary-curve}, and injectivity modulo the left $K_{\mathrm{ad}}$-action is Proposition~\ref{prop:unitary-curve-rigidity}.
\end{proof}

\begin{defn}
    We call nondegenerate integral maps in $\mathrm{Int}_{G_{\mathrm{ad}}}(S)_{K_{\mathrm{ad}}}$ \emph{unitary} curves. For a solution of the Toda system, the unitary curves in its preimage under the map $\mathrm{Int}_{G_\mathrm{ad}}(S)_{K_\mathrm{ad}}
\longrightarrow\mathrm{Toda}_\mathfrak{g}(S)$ are said to be \emph{associated} with it and are called its \emph{associated curves}.
\end{defn}
\section{Reduced Solutions of Toda Systems}\label{sec:reduced}
\subsection{Definition}

Choose Chevalley root vectors $e_i,f_i$ compatibly with the compact form, so that $[e_i,f_j]=\delta_{ij}h_i$. The inverse of an indecomposable finite-type Cartan matrix has strictly positive entries, so the numbers $2\sum_j a^{ji}$ below are positive. Define
\[
m_i:=2\sum_{j=1}^n a^{ji},\qquad
p_1=\sum_{i=1}^n\sqrt{m_i}\,e_i,\qquad
p_{-1}=\sum_{i=1}^n\sqrt{m_i}\,f_i,
\qquad h=[p_1,p_{-1}]=\sum_i m_i h_i.
\]
Then $\alpha_i(h)=2$ for every $i$, and
\[
[h,p_1]=2p_1,\qquad [h,p_{-1}]=-2p_{-1},\qquad
[p_1,p_{-1}]=h.
\]
Thus $(p_1,h,p_{-1})$ is a compact-compatible principal $\mathfrak{sl}_2$-triple \cite{Kostant1959}; this principal-embedding choice also appears in the zero-curvature construction of Razumov--Saveliev \cite[\S2.4, p.~480, and Appendix~B, Eq.~(B.42)]{RazumovSaveliev1994}. It integrates to a homomorphism $\Phi_{\mathfrak g}:\operatorname{SL}_2(\CC)\longrightarrow G$. With this compact-compatible choice, $\Phi_{\mathfrak g}(\operatorname{SU}(2))\subset K$. The descent to the adjoint groups is well defined: indeed, $\Phi_{\mathfrak g}(-I)=\exp(\pi\sqrt{-1}h)$. For every root $\beta$ one has $\beta(h)\in2\mathbb Z$, because $\alpha_i(h)=2$ for every simple root. Hence $\exp(\pi\sqrt{-1}h)$ acts trivially in the adjoint representation and therefore belongs to $Z(G)$. Thus the composite $\operatorname{SL}_2(\CC)\to G\to G_{\mathrm{ad}}$ kills $\{\pm I\}$. After passing to adjoint groups it induces
\[
\overline\Phi_{\mathfrak g}:
\operatorname{PSL}_2(\CC)\longrightarrow G_{\mathrm{ad}},
\qquad
\sigma_{\mathfrak g}:=
\overline\Phi_{\mathfrak g}|_{\operatorname{PSU}(2)}.
\]
Since $\PP^1$ is the flag variety of $\operatorname{SL}_2(\CC)$, $\Phi_{\mathfrak g}$ induces a nondegenerate integral curve
\[
r_{\mathfrak g}:\PP^1\longrightarrow F,
\]
called the \emph{principal curve}. More explicitly, fix the Borel subgroup $B_2\subset\operatorname{SL}_2(\CC)$ satisfying $\Phi_{\mathfrak g}(B_2)\subset B_0$. Then
\[
r_{\mathfrak g}(aB_2)
=\Phi_{\mathfrak g}(a)B_0\Phi_{\mathfrak g}(a)^{-1}.
\]
Indeed, its tangent at the base point is represented by $p_1\bmod\mathfrak b_0$, whose simple-root coefficients $\sqrt{m_i}$ are all nonzero.

\begin{defn}\label{def:reduced-solution}
Let $\boldsymbol\varphi=(\varphi_1,\ldots,\varphi_n)$ be a solution of the $\mathfrak g$-Toda system on a Riemann surface $S$. We call $\boldsymbol\varphi$ a \emph{reduced solution} if it is generated by a solution $\omega$ of the $\operatorname{SU}(2)$ Toda system in the following sense: there exist a holomorphic $\operatorname{PSU}(2)$-equivariant developing pair $v=(\widetilde v,\rho_v)$ associated with $\omega$ and an element $g\in G_{\mathrm{ad}}$ such that a $K_{\mathrm{ad}}$-unitary curve associated with $\boldsymbol\varphi$ has the form
\[
\widetilde f=g\cdot r_{\mathfrak g}\circ\widetilde v:
\widetilde S\longrightarrow F.
\]
Thus reduced solutions are precisely those arising from $\operatorname{SU}(2)$ Toda solutions through the principal $\operatorname{SL}_2$ subgroup of $G$.
\end{defn}

\subsection{The reduced solution generated by an
\texorpdfstring{$\operatorname{SU}(2)$}{SU(2)} solution}

From now on, the notation $v:S\to\PP^1$ abbreviates a holomorphic unitary equivariant developing pair
\[
\widetilde v:\widetilde S\to\PP^1,\qquad
\rho_v:\pi_1(S)\to\operatorname{PSU}(2)
\]
with monodromy in the displayed fixed compact group. Similarly, a unitary flag curve has monodromy contained in the chosen $K_{\mathrm{ad}}$. When it represents a smooth positive Toda metric, $\widetilde v$ is locally univalent. In the singular setting it remains locally univalent on the punctured surface; if it extends holomorphically across a deleted cone point, that extension may be branched at the cone point.

Denote $C=(a_{ij})$ for the Cartan matrix of $\mathfrak g$ and $C^{-1}=(a^{ij})$, and recall that $\varpi_i$ denotes the $i$-th fundamental weight. Define
\[
m_i=\varpi_i(h),\qquad i=1,\ldots,n,
\]
where $h$ is the semisimple element of the principal $\mathfrak{sl}_2$-triple. Since $\alpha_j(h)=2$ for every simple root $\alpha_j$, one has
\begin{equation}\label{eq:principal-weight-coefficients}
(m_1,\ldots,m_n)C=(2,\ldots,2),
\qquad
m_i=2\sum_{j=1}^n a^{ji}.
\end{equation}

\begin{thm}\label{thm:principal-reduced-solution}
Let $\omega$ be a smooth positive solution of the $\operatorname{SU}(2)$ Toda system on a Riemann surface $S$, and let $v:S\to\PP^1$ be its locally univalent unitary developing map, so that $\omega=v^*\varphi_{\mathrm{FS}}$. Then
\[
\boldsymbol\omega_{\mathfrak g}
:=(m_1\omega,\ldots,m_n\omega)
\]
is a reduced solution of the $\mathfrak g$-Toda system. The principal curve
\[
f=r_{\mathfrak g}\circ v:S\longrightarrow F
\]
is unitary and associated with this solution, and, for every $i=1,\ldots,n$,
\begin{equation}\label{eq:principal-curve-pullback}
f^*\pi_i^*\varphi_{\mathrm{FS}}
=m_i\omega
=\varpi_i(h)\omega
=2\left(\sum_{j=1}^n a^{ji}\right)\omega.
\end{equation}
Moreover, the associated curves of $\boldsymbol\omega_{\mathfrak g}$ are obtained from $f$ by the action of $K_{\mathrm{ad}}$.
\end{thm}

\begin{proof}
We first prove that $f=r_{\mathfrak g}\circ v$ has unitary monodromy. Let $p:\widetilde S\to S$ be the universal covering and choose a branch
\[
\widetilde v:\widetilde S\longrightarrow\PP^1
\]
of the developing map of $\omega$. Since $v$ is unitary, there is a monodromy representation
\[
\rho_v:\pi_1(S)\longrightarrow\operatorname{PSU}(2)
\]
such that
\[
\widetilde v(\gamma\cdot z)
=\rho_v(\gamma)\cdot\widetilde v(z)
\]
for all $\gamma\in\pi_1(S)$ and $z\in\widetilde S$.

The adjoint homomorphism $\overline\Phi_{\mathfrak g}$ defined above maps $\operatorname{PSU}(2)$ into $K_{\mathrm{ad}}$. The principal curve is equivariant with respect to this homomorphism:
\[
r_{\mathfrak g}(a\cdot x)
=\overline\Phi_{\mathfrak g}(a)\cdot r_{\mathfrak g}(x).
\]
Therefore, if $\widetilde f=r_{\mathfrak g}\circ\widetilde v$, then
\[
\begin{aligned}
\widetilde f(\gamma\cdot z)
&=r_{\mathfrak g}\bigl(\rho_v(\gamma)\cdot\widetilde v(z)\bigr)\\
&=\overline\Phi_{\mathfrak g}\bigl(\rho_v(\gamma)\bigr)
  \cdot\widetilde f(z).
\end{aligned}
\]
Thus the monodromy representation of $f$ is
\[
\rho_f=\overline\Phi_{\mathfrak g}\circ\rho_v:
\pi_1(S)\longrightarrow K_{\mathrm{ad}}.
\]
This proves that $r_{\mathfrak g}\circ v$ is a holomorphic unitary integral curve.

We next compute the Toda forms induced by this curve. Let $[v_i]=\pi_i(B_0)\in\PP(V_i)$ be the evaluation line. By the Borel--Weil convention fixed above, this line has $B_0$-character $\omega_i^{-1}$ and hence infinitesimal weight $-\varpi_i$. Therefore
\[
p_{-1}v_i=0,
\qquad
hv_i=-\varpi_i(h)v_i=-m_iv_i.
\]
The principal $\mathfrak{sl}_2$-module generated by $v_i$ therefore has highest weight $m_i$. Consequently, the orbit of $[v_i]$ is the degree-$m_i$ Veronese curve in the projectivization of the irreducible $\operatorname{SL}_2(\CC)$-module generated by $v_i$. With the invariant Hermitian metrics used to define the Fubini--Study forms, this gives
\[
r_{\mathfrak g}^*\pi_i^*\varphi_{\mathrm{FS}}=m_i\varphi_{\mathrm{FS}}.
\]
Pulling back by $v$ yields
\[
(r_{\mathfrak g}\circ v)^*\pi_i^*\varphi_{\mathrm{FS}}
=m_i v^*\varphi_{\mathrm{FS}}=m_i\omega,
\]
which proves \eqref{eq:principal-curve-pullback}.

The last statement follows from Proposition~\ref{prop:unitary-curve-rigidity}.
\end{proof}

\begin{rmk}
For $\mathfrak g=\mathfrak{sl}_{n+1}(\CC)$, formula \eqref{eq:principal-weight-coefficients} gives
\[
m_i=i(n+1-i),\qquad i=1,\ldots,n.
\]
Hence Theorem~\ref{thm:principal-reduced-solution} recovers
\[
\bigl(i(n+1-i)\omega\bigr)_{i=1}^n,
\]
the reduced $\operatorname{SU}(n+1)$ Toda solution generated by the spherical metric $\omega$ in \cite[Theorem~1.1]{ShiWeiXu2026}.
\end{rmk}

\subsection{Classification of reduced solutions}

Let $v:S\to\PP^1$ be a holomorphic unitary developing map, and denote its monodromy image and its closure by
\[
\Gamma_v:=\rho_v\bigl(\pi_1(S)\bigr),
\qquad
G_v:=\overline{\Gamma_v}\subset\operatorname{PSU}(2).
\]
Recall the homomorphism $\sigma_{\mathfrak g}:\operatorname{PSU}(2)\to K_{\mathrm{ad}}$ induced from the principal $\operatorname{PSL}_2(\CC)\to G_{\mathrm{ad}}$ homomorphism. We define
\[
\mathcal M_v^{\mathfrak g}
:=\left\{g\in G_{\mathrm{ad}}:
\text{the natural monodromy of }g\cdot r_{\mathfrak g}\circ v
\text{ lies in the fixed }
K_{\mathrm{ad}}
\right\}.
\]
Let $\tau:G_{\mathrm{ad}}\to G_{\mathrm{ad}}$ be the antiholomorphic involution whose fixed-point group is $K_{\mathrm{ad}}$, and denote
\[
g^\dagger:=\tau(g)^{-1}.
\]
For a subset $A\subset G_{\mathrm{ad}}$, denote $Z_{G_{\mathrm{ad}}}(A)$ for its centralizer.

For certain type-$A$ systems on the three-punctured sphere, Lin--Nie--Wei construct solutions from hypergeometric equations by choosing a positive-definite Hermitian form invariant under monodromy \mbox{\cite[Section~3]{LinNieWei2018}}. The following criterion expresses the corresponding invariant-metric condition for transformations of a fixed principal curve, for an arbitrary simple group and monodromy closure.
\begin{lem}\label{lem:unitary-transformations}
The set of transformations preserving unitary monodromy is
\begin{equation}\label{eq:Mv-centralizer}
\begin{aligned}
\mathcal M_v^{\mathfrak g}
&=\left\{g\in G_{\mathrm{ad}}:
g\sigma_{\mathfrak g}(G_v)g^{-1}\subset K_{\mathrm{ad}}\right\}\\
&=\left\{g\in G_{\mathrm{ad}}:
g^\dagger g\in
Z_{G_{\mathrm{ad}}}\bigl(\sigma_{\mathfrak g}(G_v)\bigr)\right\}.
\end{aligned}
\end{equation}
Furthermore:
\begin{enumerate}
\item for $U\in\operatorname{PSU}(2)$, $\mathcal M_{U\cdot v}^{\mathfrak g}
    =\mathcal M_v^{\mathfrak g}\,\sigma_{\mathfrak g}(U)^{-1}$;
\item if two unitary curves $v_1$ and $v_2$ have the same monodromy closure, then $\mathcal M_{v_1}^{\mathfrak g}=\mathcal M_{v_2}^{\mathfrak g}$.
\end{enumerate}
\end{lem}

\begin{proof}
By Theorem~\ref{thm:principal-reduced-solution}, the monodromy image of $r_{\mathfrak g}\circ v$ is $\sigma_{\mathfrak g}(\Gamma_v)$ and its closure is $\sigma_{\mathfrak g}(G_v)$. Analytic continuation of a branch of $g\cdot r_{\mathfrak g}\circ v$ shows that its monodromy image and closure are, respectively,
\[
g\sigma_{\mathfrak g}(\Gamma_v)g^{-1},
\qquad
g\sigma_{\mathfrak g}(G_v)g^{-1}.
\]
Since $K_{\mathrm{ad}}$ is closed, containment of the image is equivalent to containment of its closure. This proves the first equality in \eqref{eq:Mv-centralizer}.

Let $a\in\sigma_{\mathfrak g}(G_v)\subset K_{\mathrm{ad}}$. Since $\tau(a)=a$, the condition $gag^{-1}\in K_{\mathrm{ad}}$ is equivalent to
\[
\tau(gag^{-1})=gag^{-1}.
\]
Expanding the left-hand side and rearranging gives
\[
a\tau(g)^{-1}g=\tau(g)^{-1}ga,
\]
or, equivalently,
\[
a(g^\dagger g)=(g^\dagger g)a.
\]
Requiring this for every $a\in\sigma_{\mathfrak g}(G_v)$ proves the second equality.

For $U\in\operatorname{PSU}(2)$, equivariance of the principal curve gives
\[
r_{\mathfrak g}\circ(U\cdot v)
=\sigma_{\mathfrak g}(U)\cdot r_{\mathfrak g}\circ v.
\]
Hence $g\in\mathcal M_{U\cdot v}^{\mathfrak g}$ if and only if $g\sigma_{\mathfrak g}(U)\in\mathcal M_v^{\mathfrak g}$, proving (i). The centralizer of a subgroup equals the centralizer of its closure, so $\mathcal M_v^{\mathfrak g}$ depends only on $G_v$. Thus, statement (ii) follows immediately from \eqref{eq:Mv-centralizer}.
\end{proof}

Choose an Iwasawa decomposition
\[
G_{\mathrm{ad}}=K_{\mathrm{ad}}AN
\]
compatible with the Borel subgroup used to define $F$, with $A=\exp\bigl(\sqrt{-1}\operatorname{Lie}(T_{\mathrm{ad}}^c)\bigr)$ inside the Cartan subgroup determined by $B_{\mathrm{ad}}$, and define $\Delta_{\mathfrak g}:=AN$. Thus every $g\in G_{\mathrm{ad}}$ has a unique decomposition $g=k\delta$ with $k\in K_{\mathrm{ad}}$ and $\delta\in\Delta_{\mathfrak g}$; see \cite[Ch.~VII]{Knapp2002} for the global Iwasawa decomposition and its uniqueness.

\begin{lem}\label{lem:Mv-iwasawa}
Define
\[
\Delta_v^{\mathfrak g}
:=\left\{\delta\in\Delta_{\mathfrak g}:
\delta^\dagger\delta\in
Z_{G_{\mathrm{ad}}}\bigl(\sigma_{\mathfrak g}(G_v)\bigr)\right\}.
\]
Then
\[
\mathcal M_v^{\mathfrak g}
=K_{\mathrm{ad}}\Delta_v^{\mathfrak g}.
\]
\end{lem}

\begin{proof}
Denote $g=k\delta$ according to the Iwasawa decomposition. Since $k\in K_{\mathrm{ad}}$, one has $k^\dagger=k^{-1}$, and hence
\[
g^\dagger g
=(k\delta)^\dagger(k\delta)
=\delta^\dagger k^\dagger k\delta
=\delta^\dagger\delta.
\]
Lemma~\ref{lem:unitary-transformations} therefore implies that $g\in\mathcal M_v^{\mathfrak g}$ if and only if $\delta\in\Delta_v^{\mathfrak g}$, which proves the asserted decomposition.
\end{proof}

Let
\[
H_v:=\sigma_{\mathfrak g}(G_v)\subset K_{\mathrm{ad}}
\]
and denote its centralizer in the Lie algebra by
\[
\mathfrak z_{\mathfrak g}(H_v)
:=\left\{X\in\mathfrak g:
\operatorname{Ad}_aX=X\text{ for every }a\in H_v\right\}.
\]

\begin{lem}\label{lem:dimension-delta-v}
The space $\Delta_v^{\mathfrak g}$ is a real submanifold of $\Delta_{\mathfrak g}=AN$, and
\begin{equation}\label{eq:dimension-delta-v}
\dim_{\RR}\Delta_v^{\mathfrak g}
=\dim_{\CC}\mathfrak z_{\mathfrak g}(H_v).
\end{equation}
\end{lem}

\begin{proof}
Consider the positive symmetric space
\[
\mathcal P:=\exp(\ii\mathfrak k_{\mathrm{ad}})\subset G_{\mathrm{ad}},
\]
where $\mathfrak k_{\mathrm{ad}}=\mathfrak k$ is the Lie algebra of $K_{\mathrm{ad}}$. We use the standard polar decomposition and fixed-point description for a compact subgroup; see, for example, \cite[Chs.~VII--VIII]{Knapp2002}. The group-valued Cholesky map
\[
\begin{aligned}
q:\Delta_{\mathfrak g}=AN&\longrightarrow\mathcal P,\\
\delta&\longmapsto\delta^\dagger\delta
\end{aligned}
\]
is a real-analytic diffeomorphism. Indeed, given $p\in\mathcal P$, let $s=p^{1/2}\in\mathcal P$ and denote its Iwasawa decomposition as $s=k\delta$. Since $s^\dagger=s$ and $k^\dagger k=e$, one has $p=s^\dagger s=\delta^\dagger\delta$, so $q$ is surjective. If $q(\delta_1)=q(\delta_2)$ and $u:=\delta_2\delta_1^{-1}$, then $u^\dagger u=e$. Thus $u\in K_{\mathrm{ad}}$; but also $u\in AN$, and Iwasawa uniqueness gives $K_{\mathrm{ad}}\cap AN=\{e\}$. Hence $\delta_1=\delta_2$. Finally, the positive square-root map and the Iwasawa projection to the $AN$-factor are real analytic, and their composite $p\mapsto\operatorname{pr}_{AN}(p^{1/2})$ is the inverse of $q$. By the definition of $\Delta_v^{\mathfrak g}$, its restriction gives a diffeomorphism
\[
\Delta_v^{\mathfrak g}
\longrightarrow
\mathcal P\cap Z_{G_{\mathrm{ad}}}(H_v).
\]

Because $H_v$ is contained in the compact group $K_{\mathrm{ad}}$, its centralizer is invariant under the compact conjugation. Moreover, $\exp:\ii\mathfrak k_{\mathrm{ad}}\to\mathcal P$ is a diffeomorphism and is equivariant under conjugation by $K_{\mathrm{ad}}$. The fixed-point set of $H_v$ in $\mathcal P$ is therefore
\[
\mathcal P\cap Z_{G_{\mathrm{ad}}}(H_v)
=\exp\left(
\ii\mathfrak k_{\mathrm{ad}}\cap
\mathfrak z_{\mathfrak g}(H_v)
\right).
\]
It follows that
\[
\dim_{\RR}\Delta_v^{\mathfrak g}
=\dim_{\RR}\left(
\ii\mathfrak k_{\mathrm{ad}}\cap
\mathfrak z_{\mathfrak g}(H_v)
\right).
\]

Finally, $\mathfrak z_{\mathfrak g}(H_v)$ is the complexification of its compact real form
\[
\mathfrak z_{\mathfrak k_{\mathrm{ad}}}(H_v)
=\mathfrak k_{\mathrm{ad}}\cap\mathfrak z_{\mathfrak g}(H_v).
\]
Consequently,
\[
\ii\mathfrak k_{\mathrm{ad}}\cap\mathfrak z_{\mathfrak g}(H_v)
=\ii\mathfrak z_{\mathfrak k_{\mathrm{ad}}}(H_v)
\]
and
\[
\dim_{\RR}\left(\ii\mathfrak z_{\mathfrak k_{\mathrm{ad}}}(H_v)\right)
=\dim_{\CC}\mathfrak z_{\mathfrak g}(H_v),
\]
which proves \eqref{eq:dimension-delta-v}.
\end{proof}

\begin{rmk}
The principal homomorphism makes $\mathfrak g$ into a $G_v$-module by $a\cdot X:=\operatorname{Ad}_{\sigma_{\mathfrak g}(a)}X$. The centralizer in Lemma~\ref{lem:dimension-delta-v} is precisely the invariant subspace
\begin{equation}\label{eq:centralizer-as-invariants}
\mathfrak z_{\mathfrak g}(H_v)
=\mathfrak g^{G_v}
=\operatorname{Hom}_{G_v}(\mathbf 1,\mathfrak g),
\end{equation}
where $\mathbf 1$ is the trivial representation. Hence
\begin{equation}\label{eq:dimension-delta-representation}
\dim_{\RR}\Delta_v^{\mathfrak g}
=\dim_{\CC}\operatorname{Hom}_{G_v}(\mathbf 1,\mathfrak g).
\end{equation}

Let $d_1,\ldots,d_n$ be the exponents of $\mathfrak g$. Kostant's principal-subgroup theorem \cite{Kostant1959} gives
\[
\mathfrak g\simeq
\bigoplus_{j=1}^n\operatorname{Sym}^{2d_j}(\CC^2).
\]
Restricting this representation to $G_v\subset\operatorname{PSU}(2)$ gives
\[
\mathfrak z_{\mathfrak g}(H_v)
\simeq\bigoplus_{j=1}^n
\left(\operatorname{Sym}^{2d_j}(\CC^2)\right)^{G_v},
\]
and therefore
\begin{equation}\label{eq:dimension-by-exponents}
\dim_{\RR}\Delta_v^{\mathfrak g}
=\sum_{j=1}^n\dim_{\CC}
\left(\operatorname{Sym}^{2d_j}(\CC^2)\right)^{G_v}.
\end{equation}
The even symmetric powers descend from $\operatorname{SU}(2)$ to $\operatorname{PSU}(2)$, so these are representations of $G_v$ itself.
\end{rmk}
\begin{lem}\label{lem:principal-curve-stabilizer}
Let $f:\tilde{S}\to F$ be a nondegenerate integral curve. For $g\in G_\mathrm{ad}$, if $g\cdot f\equiv f$, then $g=e$ is the identity.
\end{lem}
\begin{proof}
It is enough to work on a coordinate disc $U\subset\widetilde S$. Choose a holomorphic lift $F:U\to G_{\mathrm{ad}}$ of $f$. Integrality and nondegeneracy give
\[
F^{-1}F_z=p(z)+b(z),
\qquad
p(z)=\sum_{i=1}^n c_i(z)e_i,
\qquad
c_i(z)\ne0,\qquad b(z)\in\mathfrak b.
\]
Use the principal grading determined by $h$,
\[
\mathfrak g=\bigoplus_{j\in\mathbb Z}\mathfrak g_j,
\qquad
\mathfrak g_j:=\{X\in\mathfrak g:[h,X]=2jX\},
\qquad
\mathfrak b=\bigoplus_{j\le0}\mathfrak g_j.
\]
Every $p(z)$ is conjugate by the Cartan subgroup to the principal nilpotent $p_1$. Denote the projection $\mathfrak g\to\mathfrak g_j$ by $X\mapsto X^{[j]}$. Kostant's principal-$\mathfrak{sl}_2$ decomposition \cite{Kostant1959} therefore shows that
\[
\operatorname{ad}_{p(z)}:
\mathfrak g_j\longrightarrow\mathfrak g_{j+1}
\quad\text{is injective for every }j\le0.
\]

Let $X\in\mathfrak g$ infinitesimally fix $f$, and define $x(z)=\operatorname{Ad}_{F(z)^{-1}}X$. Then $x(z)\in\mathfrak b$ and
\[
\partial_zx=-[p(z)+b(z),x(z)].
\]
The degree-one component of this equation gives $[p(z),x^{[0]}(z)]=0$, hence $x^{[0]}=0$. If $x^{[0]},\ldots,x^{[j+1]}$ vanish for some $j<0$, the degree-$(j+1)$ component gives $[p(z),x^{[j]}(z)]=0$. Descending induction and the injectivity above therefore give $x^{[j]}=0$ for every $j\le0$. Thus the infinitesimal pointwise stabilizer of $f$ is zero.

Now suppose that $g$ fixes $f$ pointwise, and denote $g=su$ for its commuting semisimple and unipotent Jordan factors. For every $z$, the equality $g\cdot f(z)=f(z)$ and self-normality of Borel subgroups imply $g\in B_z:=f(z)$. Since Jordan decomposition is internal to the algebraic group $B_z$, both $s$ and $u$ belong to every $B_z$ and therefore fix $f$ pointwise. Fix $z$. The action of $s$ on the simple-root line $\mathcal N_i(B_z)$ is multiplication by the corresponding simple-root character. Because $s$ fixes $df_z$ and every component $\operatorname{pr}_i(df_z)$ is nonzero, all these characters take the value one on the image of $s$ in $B_z/R_u(B_z)$. The simple roots form a basis of the character lattice of the adjoint Cartan, so that image is the identity. Thus $s\in R_u(B_z)$; being both semisimple and unipotent, $s=e$. Hence $g=u$ is unipotent. Moreover, $u\in B_z$ for every $z$ implies $\log u\in\operatorname{Lie}(R_u(B_z))\subset\mathfrak b_z$ for every $z$, so $\exp(t\log u)$ fixes $f$ pointwise. Therefore $\log u$ belongs to the infinitesimal pointwise stabilizer of $f$, which the preceding argument shows is zero. Thus $u=e$ and hence $g=e$.
\end{proof}

\begin{thm}
\label{thm:classification-reduced-solutions}
Let $\omega$ be a smooth positive solution of the $\operatorname{SU}(2)$ Toda system on $S$ with locally univalent unitary developing map $v:S\to\PP^1$, and let $\Gamma_v:=\rho_v(\pi_1(S))$ be its monodromy image and $G_v:=\overline{\Gamma_v}\subset\operatorname{PSU}(2)$ its closure. Then the reduced $\mathfrak g$-Toda solutions generated by $\omega$ are parametrized bijectively by
\[
\Delta_v^{\mathfrak g}
=\left\{\delta\in AN:
\delta^\dagger\delta\in
Z_{G_{\mathrm{ad}}}
\bigl(\sigma_{\mathfrak g}(G_v)\bigr)
\right\}.
\]
More precisely, the parameter $\delta$ corresponds to the solution whose associated unitary curve is $f_\delta=\delta\cdot r_{\mathfrak g}\circ v$, and every reduced solution generated by $\omega$ is obtained uniquely in this way. The parameter space has dimension
\begin{equation}\label{eq:classification-dimension}
\dim_{\RR}\Delta_v^{\mathfrak g}
=\sum_{j=1}^n
\dim_{\CC}
\left(\operatorname{Sym}^{2d_j}(\CC^2)
\right)^{G_v},
\end{equation}
where $d_1,\ldots,d_n$ are the exponents of $\mathfrak g$.

Let
\[
T_0:=\operatorname{image}\bigl(B_2\cap\operatorname{SU}(2)
\longrightarrow\operatorname{PSU}(2)\bigr)
\]
be the maximal torus determined by the Borel $B_2$ used to define the principal curve, and denote its normalizer by $N_{\operatorname{PSU}(2)}(T_0)$. In particular, the following cases hold.
\begin{enumerate}
\item If $G_v=\operatorname{PSU}(2)$, then
    \[
    \Delta_v^{\mathfrak g}=\{e\}.
    \]
Thus $\omega$ generates only the principal reduced solution, and its associated curves differ by the action of $K_{\mathrm{ad}}$.

\item If $G_v=T_0$ is a maximal torus of $\operatorname{PSU}(2)$, then
    \[
    \Delta_v^{\mathfrak g}=A,
    \qquad
    \dim_{\RR}\Delta_v^{\mathfrak g}=n.
    \]

\item If $G_v=N_{\operatorname{PSU}(2)}(T_0)$ is the normalizer of a maximal torus, then
    \[
    \dim_{\RR}\Delta_v^{\mathfrak g}
    =\#\{j: d_j\text{ is even}\}.
    \]

\item If $G_v$ is finite, and $\chi_{2d_j}$ is the character of the $\operatorname{PSU}(2)$-representation $\operatorname{Sym}^{2d_j}(\CC^2)$, then
    \begin{equation}\label{eq:finite-monodromy-dimension}
    \dim_{\RR}\Delta_v^{\mathfrak g}
    =\frac{1}{|G_v|}
    \sum_{a\in G_v}\sum_{j=1}^n\chi_{2d_j}(a).
    \end{equation}
This is the standard character projector onto the invariant subspace \cite[Ch.~2]{Serre1977}. It applies to the cyclic, dihedral, tetrahedral, octahedral, and icosahedral projective monodromy groups.
\end{enumerate}
\end{thm}
\begin{proof}
Let $\boldsymbol\varphi$ be a reduced solution generated by $\omega$. By Definition~\ref{def:reduced-solution}, it has an associated unitary curve $g\cdot r_{\mathfrak g}\circ v'$ for some $g\in G_{\mathrm{ad}}$ and some associated unitary curve $v'$ of $\omega$. Proposition~\ref{prop:unitary-curve-rigidity} gives $U\in\operatorname{PSU}(2)$ such that $v'=U\cdot v$. By the equivariance of the principal curve,
\[
g\cdot r_{\mathfrak g}\circ v'
=g\sigma_{\mathfrak g}(U)\cdot r_{\mathfrak g}\circ v.
\]
Replacing $g$ by $g\sigma_{\mathfrak g}(U)$, we may therefore use the specified associated curve $v$. Denote $\widetilde f=g\cdot r_{\mathfrak g}\circ\widetilde v$. Since $\widetilde f$ is unitary, choose an equivariance representation $\rho_{\widetilde f}:\pi_1(S)\to K_{\mathrm{ad}}$. Equivariance of the principal curve also gives the natural representation $\eta_g(\gamma):=g\sigma_{\mathfrak g}(\rho_v(\gamma))g^{-1}$. Thus, for every $\gamma\in\pi_1(S)$ and $z\in\widetilde S$, one has $\rho_{\widetilde f}(\gamma)\cdot\widetilde f(z)=\widetilde f(\gamma z)=\eta_g(\gamma)\cdot\widetilde f(z)$. Consequently, $\rho_{\widetilde f}(\gamma)^{-1}\eta_g(\gamma)$ fixes $\widetilde f$ pointwise. Lemma~\ref{lem:principal-curve-stabilizer} shows that it is the identity. Hence $\eta_g=\rho_{\widetilde f}$ has image in $K_{\mathrm{ad}}$, and therefore $g\in\mathcal M_v^{\mathfrak g}$. Lemma~\ref{lem:Mv-iwasawa} gives
\[
g=k\delta,
\qquad
k\in K_{\mathrm{ad}},\quad
\delta\in\Delta_v^{\mathfrak g}.
\]
The factor $k$ does not change the induced Toda solution, because the Fubini--Study forms are $K_{\mathrm{ad}}$-invariant. Hence every reduced solution generated by $\omega$ is represented by $\delta\cdot r_{\mathfrak g}\circ v$. Conversely, the definition of $\Delta_v^{\mathfrak g}$ and Lemma~\ref{lem:unitary-transformations} show that each such curve has unitary monodromy and therefore produces a reduced solution.

For uniqueness, suppose $\delta_1,\delta_2\in\Delta_v^{\mathfrak g}$ induce the same Toda solution. Proposition~\ref{prop:unitary-curve-rigidity} gives $k\in K_{\mathrm{ad}}$ such that
\[
\delta_2\cdot r_{\mathfrak g}\circ\widetilde v
=k\delta_1\cdot r_{\mathfrak g}\circ\widetilde v.
\]
The map $\widetilde v$ is nonconstant, so the open mapping theorem shows that its image contains a nonempty open subset of $\PP^1$. Lemma~\ref{lem:principal-curve-stabilizer} gives $\delta_2^{-1}k\delta_1=e$. Thus $\delta_2=k\delta_1$; uniqueness of the Iwasawa decomposition, with $\delta_1,\delta_2\in AN$, forces $k=e$ and $\delta_1=\delta_2$.

Formula \eqref{eq:classification-dimension} follows from \eqref{eq:dimension-by-exponents}, which is a consequence of Lemma~\ref{lem:dimension-delta-v}.

If $G_v=\operatorname{PSU}(2)$, none of the nontrivial irreducible modules $\operatorname{Sym}^{2d_j}(\CC^2)$ contains a fixed vector. Thus the dimension is zero. The positive fixed-point space is connected and contains the identity, so $\Delta_v^{\mathfrak g}=\{e\}$.

If $G_v=T_0$, every $\operatorname{Sym}^{2d_j}(\CC^2)$ has a one-dimensional zero-weight space. Thus the dimension is $n$. Moreover, $h$ is regular because $\alpha_i(h)=2$ for every $i$, so the centralizer of the principal circle is the Cartan subgroup determined by $B_{\mathrm{ad}}$. Compatibility of the Iwasawa decomposition gives $\mathcal P\cap Z_{G_{\mathrm{ad}}}(\sigma_{\mathfrak g}(T_0))=A$. Since $q(a)=a^2$ for $a\in A$, the bijectivity of $q$ gives $\Delta_v^{\mathfrak g}=q^{-1}(A)=A$.

If $G_v=N_{\operatorname{PSU}(2)}(T_0)$, the nonidentity component of the normalizer acts on the zero-weight line in $\operatorname{Sym}^{2d_j}(\CC^2)$ by $(-1)^{d_j}$. Hence that line is fixed precisely when $d_j$ is even. 

Finally, if $G_v$ is finite, averaging the character computes the dimension of the invariant subspace in each summand, which gives \eqref{eq:finite-monodromy-dimension}.
\end{proof}

\begin{rmk}
For $\mathfrak g=\mathfrak{sl}_{n+1}(\CC)$, the exponents are $1,\ldots,n$. Theorem~\ref{thm:classification-reduced-solutions} therefore specializes to Theorem~1.3 of \cite{ShiWeiXu2026}; the explicit matrix descriptions there are the type-$A_n$ realization of the intrinsic centralizer condition above.
\end{rmk}

\begin{rmk}
When $G_{\mathrm{ad}}=\operatorname{PSL}(n+1,\CC)$, $g^\dagger$ is the conjugate transpose $g^*$ and $\Delta_{\mathfrak g}$ is represented by upper-triangular matrices with positive diagonal entries. Thus Lemmas~\ref{lem:unitary-transformations} and~\ref{lem:Mv-iwasawa} reduce precisely to Lemmas~3.1 and~3.2 of \cite{ShiWeiXu2026}.

Moreover, define
\[
V_v:=\left\{A\in\operatorname{Mat}_{n+1}(\CC):
A\widetilde a=\widetilde aA\text{ for every }a\in H_v\right\},
\]
where $\widetilde a\in\operatorname{SL}_{n+1}(\CC)$ is any lift of the projective class $a$. The commuting condition is independent of the scalar choices of these lifts. Then
\[
\mathfrak z_{\mathfrak{sl}_{n+1}}(H_v)
=V_v\cap\mathfrak{sl}_{n+1}(\CC).
\]
Since $V_v$ contains the scalar matrices, this space has complex codimension one in $V_v$. Formula~\eqref{eq:dimension-delta-v} therefore becomes
\[
\dim_{\RR}\Delta_v^{\mathfrak{sl}_{n+1}}
=\dim_{\CC}V_v-1,
\]
which is Lemma~3.6 of \cite{ShiWeiXu2026}.
\end{rmk}

For comparison, specialize the one-source type-$A_n$ classification of Lin--Wei--Ye \mbox{\cite{LinWeiYe2012}} to equal source strengths $\gamma_i=\beta-1$, with $\beta>0$. Their formulas give projective transforms of $[1:z^\beta:\cdots:z^{n\beta}]$, up to fixed coordinate normalizations, with parameters restricted by single-valuedness. This is the reduced family generated by $v=z^\beta$ on $\mathbb C^*$; its monodromy closure is finite cyclic for rational $\beta$ and a maximal torus for irrational $\beta$.
\section{Type-by-type dimension formulas}\label{sec:classification}
Theorem~\ref{thm:classification-reduced-solutions} already reduces the problem to fixed vectors in the principal $\mathfrak{sl}_2$-decomposition. We first compute the contribution of one summand and then sum over the exponent multiset of the Lie algebra. Thus the formulas below are specializations of the master theorem, not additional local Toda--flag classification statements.
Throughout this section, $T_0$ is the compatible maximal torus from Theorem~\ref{thm:classification-reduced-solutions}, and $N(T_0):=N_{\operatorname{PSU}(2)}(T_0)$. We use projective rotation groups: $|C_q|=q$ $(q\geq1)$, $|D_q|=2q$ $(q\geq2)$, and $(|\mathsf T|,|\mathsf O|,|\mathsf I|)=(12,24,60)$. The inverse images of these finite groups in $\operatorname{SU}(2)$ are the corresponding binary groups. Under $\operatorname{PSU}(2)\cong\operatorname{SO}(3)$ these, together with $\operatorname{PSU}(2)$, $T_0$, and $N(T_0)\simeq\operatorname{O}(2)$, exhaust the closed subgroup cases; see \cite[Theorem~1.3]{ShiWeiXu2026}. We distinguish the Lie type $D_n$ from the dihedral group $D_q$.

After conjugating the developing pair, we may place a cyclic group in $T_0$ and $D_q$ in $N(T_0)$; this preserves the spherical metric and the classification by Lemma~\ref{lem:unitary-transformations}.

Fix a smooth positive $\operatorname{SU}(2)$ Toda solution on a Riemann surface $S$ with a locally univalent unitary developing pair $v=(\widetilde v,\rho_v)$, where $\widetilde v:\widetilde S\to\PP^1$, and define $G_v:=\overline{\rho_v(\pi_1(S))}$. Every $\Delta_v^X$ below is the centralizer slice from Theorem~\ref{thm:classification-reduced-solutions}; only its dimension is being specialized. Let $\mathcal E_X$ denote the exponent multiset \cite{Kostant1959}.
{
\begin{equation}
\label{eq:section3-exponents}
\begin{array}{c|l}
X&\mathcal E_X\\ \hline
A_n&1,2,\ldots,n\\
B_n&1,3,\ldots,2n-1\\
C_n&1,3,\ldots,2n-1\\
D_n&1,3,\ldots,2n-3,\ n-1\\
E_6&1,4,5,7,8,11\\
E_7&1,5,7,9,11,13,17\\
E_8&1,7,11,13,17,19,23,29\\
F_4&1,5,7,11\\
G_2&1,5
\end{array}
\end{equation}}
The $D_n$ row is a multiset: when $n$ is even, $n-1$ occurs twice. Set $S_X(q):=\sum_{d\in\mathcal E_X}\lfloor d/q\rfloor$.

For use in each of the polyhedral cases, denote
\[
m_G(d):=\dim_{\CC}\left(\operatorname{Sym}^{2d}(\CC^2)\right)^G,
\qquad G\in\{\mathsf T,\mathsf O,\mathsf I\},
\]
and, for positive integers $a,b$, denote
\[
c_{a,b}(\ell)
:=\#\{(u,w)\in\ZZ_{\geq0}^2:au+bw=\ell\},
\qquad c_{a,b}(\ell):=0\quad(\ell<0).
\]

\begin{prop}
\label{prop:explicit-polyhedral-dimensions}
The multiplicities $m_G(d)$ satisfy
\begin{equation}
\label{eq:polyhedral-molien-series}
\begin{aligned}
\sum_{d\geq0}m_T(d)t^d&=\frac{1+t^6}{(1-t^3)(1-t^4)},\\
\sum_{d\geq0}m_O(d)t^d&=\frac{1+t^9}{(1-t^4)(1-t^6)},\\
\sum_{d\geq0}m_I(d)t^d&=\frac{1+t^{15}}{(1-t^6)(1-t^{10})}.
\end{aligned}
\end{equation}
Equivalently,
\begin{equation}
\label{eq:polyhedral-coefficient-formulas}
\begin{aligned}
m_T(d)&=c_{3,4}(d)+c_{3,4}(d-6),\\
m_O(d)&=c_{4,6}(d)+c_{4,6}(d-9),\\
m_I(d)&=c_{6,10}(d)+c_{6,10}(d-15).
\end{aligned}
\end{equation}

Let
\[
M_G(n):=\sum_{j=1}^n m_G(2j-1),\qquad M_G(0):=0.
\]
If $n=6k+r$, where $0\leq r<6$, then
\begin{align}
M_T(n)&=6k^2+(2r-1)k+\alpha_r,
& (\alpha_0,\ldots,\alpha_5)&=(0,0,1,1,2,4),
\label{eq:MT-closed}\\
M_O(n)&=3k^2+(r-2)k+\beta_r,
& (\beta_0,\ldots,\beta_5)&=(0,0,0,0,0,1).
\label{eq:MO-closed}
\end{align}
If $n=15k+r$, where $0\leq r<15$, then
\begin{equation}
\label{eq:MI-closed}
M_I(n)=\frac{15k^2-7k}{2}+rk+\gamma_r,
\qquad
(\gamma_0,\ldots,\gamma_{14})
=(0,0,0,0,0,0,0,0,1,1,1,2,2,3,4).
\end{equation}
\end{prop}

\begin{proof}
For a rotation through angle $\theta$, the character of $\operatorname{Sym}^{2d}(\CC^2)$ is
\[
\chi_{2d}(\theta)
=\frac{\sin((2d+1)\theta/2)}{\sin(\theta/2)},
\qquad
\sum_{d\geq0}\chi_{2d}(\theta)t^d
=\frac{1+t}{1-2t\cos\theta+t^2}.
\]
Here $\theta$ is the usual rotation angle in $\operatorname{SO}(3)$, and $\chi_{2d}(0)$ is understood by continuity as $2d+1$. Combining inverse-angle classes, the relevant conjugacy-class data are
\[
\begin{array}{c|rrrrrr}
G&0&\pi/2&2\pi/3&\pi&2\pi/5&4\pi/5\\ \hline
\mathsf T\simeq A_4&1&0&8&3&0&0\\
\mathsf O\simeq S_4&1&6&8&9&0&0\\
\mathsf I\simeq A_5&1&0&20&15&12&12
\end{array}
\]
where each entry is the number of rotations with that angle. Thus
\[
\begin{aligned}
m_T(d)&=\frac{\chi_{2d}(0)+8\chi_{2d}(2\pi/3)+3\chi_{2d}(\pi)}{12},\\
m_O(d)&=\frac{\chi_{2d}(0)+6\chi_{2d}(\pi/2)+8\chi_{2d}(2\pi/3)+9\chi_{2d}(\pi)}{24},\\
m_I(d)&=\frac{\chi_{2d}(0)+12\chi_{2d}(2\pi/5)+12\chi_{2d}(4\pi/5)+20\chi_{2d}(2\pi/3)+15\chi_{2d}(\pi)}{60}.
\end{aligned}
\]
These are the character projectors for the rotation groups (see also Meyer \cite{Meyer1954}); averaging the generating-function identity over the displayed rows gives \eqref{eq:polyhedral-molien-series}, and coefficient extraction gives \eqref{eq:polyhedral-coefficient-formulas}.

Reading the odd coefficients in \eqref{eq:polyhedral-coefficient-formulas} gives the period shifts
$m_T(2(j+6)-1)=m_T(2j-1)+2$, $m_O(2(j+6)-1)=m_O(2j-1)+1$, and
$m_I(2(j+15)-1)=m_I(2j-1)+1$. Summing complete periods and the remaining terms gives
\eqref{eq:MT-closed}--\eqref{eq:MI-closed}.
\end{proof}

{
\begin{prop}
\label{prop:fixed-space-dictionary}
For $d\geq1$, the even symmetric power $\operatorname{Sym}^{2d}(\CC^2)$ descends to $\operatorname{PSU}(2)$. For a closed subgroup $H\subset\operatorname{PSU}(2)$, denote
\[
\mu_d(H):=\dim_{\CC}\left(\operatorname{Sym}^{2d}(\CC^2)\right)^H.
\]
In the projective normalization used here, the weights are $-d,-d+1,\ldots,d$, and
\[
\begin{array}{c|c}
H&\mu_d(H)\\ \hline
\operatorname{PSU}(2)&0\\
T_0&1\\
N(T_0)&\dfrac{1+(-1)^d}{2}\\
C_q&2\left\lfloor\dfrac d q\right\rfloor+1\\
D_q&\left\lfloor\dfrac d q\right\rfloor+\dfrac{1+(-1)^d}{2}\\
P\in\{\mathsf T,\mathsf O,\mathsf I\}&m_P(d)
\end{array}
\]
where $m_P(d)$ is given by Proposition~\ref{prop:explicit-polyhedral-dimensions}.
\end{prop}
}

\begin{proof}
The $\operatorname{PSU}(2)$ row uses irreducibility; the $T_0$ and $C_q$ rows count weights; the Weyl involution acts on the zero-weight line by $(-1)^d$, which explains the normalizer row; and the $D_q$ row additionally pairs opposite weights. The last row is the character average for the projective rotation group.
\end{proof}

Consequently, for every simple type $X$, the dimension is obtained by the single substitution
{\begin{equation}
\label{eq:section3-universal-dimension}
\dim_{\RR}\Delta_v^X=\sum_{d\in\mathcal E_X}\mu_d(G_v).
\end{equation}}

\subsection{Type \texorpdfstring{$B_n,C_n$}{B-n,C-n}}
\begin{thm}
\label{thm:Bn-classification}
Under the standing notation above, for $X\in\{B_n,C_n\}$ with $n\geq2$, we have
{\begin{equation}\label{eq:Bn-basic-dimension}
\dim_{\RR}\Delta_v^{X}
=\sum_{j=1}^n\mu_{2j-1}(G_v).
\end{equation}}
The dictionary gives $\Delta_v^X=\{e\}$ for full or normalizer monodromy, and $\Delta_v^X=A$ with dimension $n$ for toral monodromy. For the finite cases,
{\begin{align}
\dim_{\RR}\Delta_v^X
&=\sum_{j=1}^n
\left(2\left\lfloor\frac{2j-1}{q}\right\rfloor+1\right),
&&G_v=C_q,                              \label{eq:Bn-cyclic}\\
\dim_{\RR}\Delta_v^X
&=\sum_{j=1}^n
\left\lfloor\frac{2j-1}{q}\right\rfloor,
&&G_v=D_q.                              \label{eq:Bn-dihedral}
\end{align}}
For $G_v=\mathsf T,\mathsf O,\mathsf I$, respectively, the dimensions are
{\begin{align}
\dim_{\RR}\Delta_v^X
&=M_T(n),
  &&G_v=\mathsf T,                              \label{eq:Bn-T}\\
\dim_{\RR}\Delta_v^X
&=M_O(n),
  &&G_v=\mathsf O,                              \label{eq:Bn-O}\\
\dim_{\RR}\Delta_v^X
&=M_I(n),
  &&G_v=\mathsf I.                              \label{eq:Bn-I}
\end{align}}
\end{thm}

\begin{proof}
Kostant's exponent list for $B_n$ is $1,3,\ldots,2n-1$ \cite{Kostant1959}. Substitution in the fixed-space dictionary proves \eqref{eq:Bn-basic-dimension} and all of the displayed finite-group formulas.

The exponent multiset of $C_n$ is identical to that of $B_n$, so the common dictionary proves all displayed formulas. For $n\geq3$ the ambient groups differ; for $n=2$ the accidental isomorphism $B_2=C_2$ identifies the constructions after interchanging the Dynkin nodes.
\end{proof}

\begin{rmk}
For $n\geq3$, although the adjoint groups and principal curves in types $B_n$ and $C_n$ are different, their classification dimensions agree for every monodromy group $G_v$. When $n=2$, the two types are identified by the accidental isomorphism $B_2=C_2$, with the Dynkin nodes interchanged. This is precisely because the two root systems have the same list of exponents.
\end{rmk}

\subsection{Type \texorpdfstring{$D_n$}{D-n}}

\begin{thm}
\label{thm:Dn-classification}
Under the standing notation, we have
{\begin{equation}\label{eq:Dn-basic-dimension}
\dim_{\RR}\Delta_v^{D_n}
=\sum_{d\in\mathcal E_{D_n}}\mu_d(G_v).
\end{equation}}
For full and toral monodromy the dimensions are $0$ and $n$, respectively, and $\Delta_v^{D_n}=A$ in the toral case. For the normalizer and finite cases,
{\begin{align}
\dim_{\RR}\Delta_v^{D_n}&=
\begin{cases}1,&n\ \text{odd},\\0,&n\ \text{even},\end{cases}
&&G_v=N(T_0),\label{eq:Dn-PO2}\\
\dim_{\RR}\Delta_v^{D_n}&=\sum_{j=1}^{n-1}\left(2\left\lfloor\frac{2j-1}{q}\right\rfloor+1\right)+2\left\lfloor\frac{n-1}{q}\right\rfloor+1,
&&G_v=C_q,\label{eq:Dn-cyclic}\\
\dim_{\RR}\Delta_v^{D_n}&=\sum_{j=1}^{n-1}\left\lfloor\frac{2j-1}{q}\right\rfloor+\left\lfloor\frac{n-1}{q}\right\rfloor+\varepsilon_n,
&&G_v=D_q,\label{eq:Dn-dihedral}
\end{align}}
where $\varepsilon_n=1$ for odd $n$ and $0$ for even $n$. For the polyhedral groups,
{\begin{align}
\dim_{\RR}\Delta_v^{D_n}&=M_T(n-1)+m_T(n-1),&&G_v=\mathsf T,\label{eq:Dn-T}\\
\dim_{\RR}\Delta_v^{D_n}&=M_O(n-1)+m_O(n-1),&&G_v=\mathsf O,\label{eq:Dn-O}\\
\dim_{\RR}\Delta_v^{D_n}&=M_I(n-1)+m_I(n-1),&&G_v=\mathsf I.\label{eq:Dn-I}
\end{align}}
\end{thm}

\begin{proof}
The exponent multiset is $\mathcal E_{D_n}=\{1,3,\ldots,2n-3,n-1\}$, with multiplicity when $n$ is even. For the normalizer and dihedral rows, the extra exponent $n-1$ contributes the parity term $\varepsilon_n$; the cyclic and polyhedral rows follow by direct substitution in the same multiset dictionary.
\end{proof}

\subsection{Types \texorpdfstring{$E_6,E_7,E_8$}{E6, E7, E8}}

\begin{thm}
\label{thm:E-classification}
Under the standing notation, for $X\in\{E_6,E_7,E_8\}$, we have
{\begin{equation}\label{eq:E-basic-dimension}
\dim_{\RR}\Delta_v^X
=\sum_{d\in\mathcal E_X}\mu_d(G_v).
\end{equation}}
For full monodromy the dimension is $0$; for toral monodromy $\Delta_v^X=A$ and the dimensions are $(6,7,8)$ for $(E_6,E_7,E_8)$. The remaining normalizer, cyclic, dihedral, and polyhedral cases are
{\begin{align}
\dim_{\RR}\Delta_v^X&=\begin{cases}2,&X=E_6,\\0,&X=E_7\text{ or }E_8,\end{cases}&&G_v=N(T_0),\label{eq:E-PO2}\\
\dim_{\RR}\Delta_v^X&=\sum_{d\in\mathcal E_X}\left(2\left\lfloor\frac d q\right\rfloor+1\right),&&G_v=C_q,\label{eq:E-cyclic}\\
\dim_{\RR}\Delta_v^X&=\sum_{d\in\mathcal E_X}\left(\left\lfloor\frac d q\right\rfloor+\frac{1+(-1)^d}{2}\right),&&G_v=D_q.\label{eq:E-dihedral}
\end{align}}
For the polyhedral groups,
{\begin{align}
\left(\dim_{\RR}\Delta_v^{E_6},\dim_{\RR}\Delta_v^{E_7},\dim_{\RR}\Delta_v^{E_8}\right)&=(4,8,16),&&G_v=\mathsf T,\label{eq:E-T}\\
\left(\dim_{\RR}\Delta_v^{E_6},\dim_{\RR}\Delta_v^{E_7},\dim_{\RR}\Delta_v^{E_8}\right)&=(2,3,6),&&G_v=\mathsf O,\label{eq:E-O}\\
\left(\dim_{\RR}\Delta_v^{E_6},\dim_{\RR}\Delta_v^{E_7},\dim_{\RR}\Delta_v^{E_8}\right)&=(0,0,0),&&G_v=\mathsf I.\label{eq:E-I}
\end{align}}
\end{thm}

\begin{proof}
Kostant's theorem \cite{Kostant1959} supplies one summand $\operatorname{Sym}^{2d}(\CC^2)$ for each $d\in\mathcal E_X$. Substitution in the fixed-space dictionary proves every displayed formula. The two even exponents $4,8$ of $E_6$ give the normalizer value $2$, while $E_7$ and $E_8$ have none; Proposition~\ref{prop:explicit-polyhedral-dimensions} gives the polyhedral values.
\end{proof}

\subsection{Type \texorpdfstring{$F_4$}{F4}}

\begin{thm}
\label{thm:F4-classification}
Under the standing notation, we have
{\begin{equation}\label{eq:F4-basic-dimension}
\dim_{\RR}\Delta_v^{F_4}
=\sum_{d\in\mathcal E_{F_4}}\mu_d(G_v).
\end{equation}}
For full and toral monodromy the dimensions are $0$ and $4$ (with $\Delta_v^{F_4}=A$ in the toral case), and the normalizer gives dimension $0$. For the finite cases,
{\begin{align}
\dim_{\RR}\Delta_v^{F_4}&=\sum_{d\in\mathcal E_{F_4}}\left(2\left\lfloor\frac d q\right\rfloor+1\right),&&G_v=C_q,\label{eq:F4-cyclic}\\
\dim_{\RR}\Delta_v^{F_4}&=\sum_{d\in\mathcal E_{F_4}}\left\lfloor\frac d q\right\rfloor,&&G_v=D_q,\label{eq:F4-dihedral}\\
\dim_{\RR}\Delta_v^{F_4}&=2,&&G_v=\mathsf T,\label{eq:F4-T}\\
\dim_{\RR}\Delta_v^{F_4}&=0,&&G_v=\mathsf O,\label{eq:F4-O}\\
\dim_{\RR}\Delta_v^{F_4}&=0,&&G_v=\mathsf I.\label{eq:F4-I}
\end{align}}
\end{thm}

\begin{proof}
Substitution of $\mathcal E_{F_4}$ in the fixed-space dictionary proves the basic, cyclic, and dihedral formulas; the polyhedral values are the corresponding entries of Proposition~\ref{prop:explicit-polyhedral-dimensions}.
\end{proof}

\subsection{Type \texorpdfstring{$G_2$}{G2}}
\begin{thm}
\label{thm:G2-classification}
Under the standing notation, we have
{\begin{equation}\label{eq:G2-basic-dimension}
\dim_{\RR}\Delta_v^{G_2}
=\sum_{d\in\mathcal E_{G_2}}\mu_d(G_v).
\end{equation}}
For full, toral, and normalizer monodromy the dimensions are $0,2,0$, with $\Delta_v^{G_2}=A$ in the toral case. The finite specializations are
{\begin{align}
\dim_{\RR}\Delta_v^{G_2}&=2\left\lfloor\frac1q\right\rfloor+2\left\lfloor\frac5q\right\rfloor+2,&&G_v=C_q,\label{eq:G2-cyclic}\\
\dim_{\RR}\Delta_v^{G_2}&=\left\lfloor\frac1q\right\rfloor+\left\lfloor\frac5q\right\rfloor,&&G_v=D_q,\label{eq:G2-dihedral}\\
\dim_{\RR}\Delta_v^{G_2}&=0,&&G_v=\mathsf T,\label{eq:G2-T}\\
\dim_{\RR}\Delta_v^{G_2}&=0,&&G_v=\mathsf O,\label{eq:G2-O}\\
\dim_{\RR}\Delta_v^{G_2}&=0,&&G_v=\mathsf I.\label{eq:G2-I}
\end{align}}
\end{thm}

\begin{proof}
Substitution of $\mathcal E_{G_2}$ in the fixed-space dictionary proves all displayed formulas; the polyhedral entries are zero by Proposition~\ref{prop:explicit-polyhedral-dimensions}.
\end{proof}

\subsection{Comparison of all simple types}

We now collect the formulas for all simple types. The exponent table and the quantity $S_X(q)$ were defined at the start of this section, so the three tables below are a compact lookup of the universal dictionary rather than new proofs.
For $G\in\{\mathsf T,\mathsf O,\mathsf I\}$, also define
\[
P_G(n):=\sum_{d=1}^n m_G(d).
\]
Thus the type-$A_n$ polyhedral dimensions are given explicitly by the coefficient formulas \eqref{eq:polyhedral-coefficient-formulas}; equivalently,
\begin{equation}
\label{eq:An-polyhedral-generating}
\sum_{n\geq0}P_G(n)t^n
=\frac{1}{1-t}\left(\sum_{d\geq0}m_G(d)t^d-1\right),
\end{equation}
where the three series on the right are given in \eqref{eq:polyhedral-molien-series}.

\begin{cor}
\label{cor:all-types-dimension-table}
For full, toral, and torus-normalizer monodromy, the dimensions are
\[
\begin{array}{c|ccc}
X&G_v=\operatorname{PSU}(2)&G_v=T_0
&G_v=N_{\operatorname{PSU}(2)}(T_0)\\ \hline
A_n&0&n&\lfloor n/2\rfloor\\
B_n&0&n&0\\
C_n&0&n&0\\
D_n&0&n&\varepsilon_n\\
E_6&0&6&2\\
E_7&0&7&0\\
E_8&0&8&0\\
F_4&0&4&0\\
G_2&0&2&0,
\end{array}
\qquad
\varepsilon_n=
\begin{cases}
1,&n\ \text{odd},\\
0,&n\ \text{even}.
\end{cases}
\]

For cyclic and orientation-preserving dihedral monodromy, the uniform comparison is
\[
\begin{array}{c|cc}
X&G_v=C_q&G_v=D_q\ \text{of order }2q\\ \hline
A_n&2S_{A_n}(q)+n&S_{A_n}(q)+\lfloor n/2\rfloor\\
B_n&2S_{B_n}(q)+n&S_{B_n}(q)\\
C_n&2S_{C_n}(q)+n&S_{C_n}(q)\\
D_n&2S_{D_n}(q)+n&S_{D_n}(q)+\varepsilon_n\\
E_6&2S_{E_6}(q)+6&S_{E_6}(q)+2\\
E_7&2S_{E_7}(q)+7&S_{E_7}(q)\\
E_8&2S_{E_8}(q)+8&S_{E_8}(q)\\
F_4&2S_{F_4}(q)+4&S_{F_4}(q)\\
G_2&2S_{G_2}(q)+2&S_{G_2}(q).
\end{array}
\]

Finally, the tetrahedral, octahedral, and icosahedral dimensions are
\begin{equation}
\label{eq:all-types-polyhedral-table}
\begin{array}{c|ccc}
X&G_v=\mathsf T&G_v=\mathsf O&G_v=\mathsf I\\ \hline
A_n&P_T(n)&P_O(n)&P_I(n)\\
B_n&M_T(n)&M_O(n)&M_I(n)\\
C_n&M_T(n)&M_O(n)&M_I(n)\\
D_n&M_T(n-1)+m_T(n-1)&M_O(n-1)+m_O(n-1)
&M_I(n-1)+m_I(n-1)\\
E_6&4&2&0\\
E_7&8&3&0\\
E_8&16&6&0\\
F_4&2&0&0\\
G_2&0&0&0.
\end{array}
\end{equation}
Here $M_T,M_O,M_I$ are the explicit quasipolynomials \eqref{eq:MT-closed}--\eqref{eq:MI-closed}, while $m_G$ is given by \eqref{eq:polyhedral-coefficient-formulas} and $P_G$ by \eqref{eq:An-polyhedral-generating}.
\end{cor}

\begin{proof}
Apply the universal formula \eqref{eq:section3-universal-dimension} to the exponent table. The first two tables use the $\operatorname{PSU}(2)$, toral, normalizer, cyclic, and dihedral rows of Proposition~\ref{prop:fixed-space-dictionary}; the last table uses its polyhedral row together with Proposition~\ref{prop:explicit-polyhedral-dimensions}.
\end{proof}

\section{Applications}\label{sec:application}

We apply the classification results to Toda systems with cone singularities. The type-$A$ construction appears in \cite[Corollary~1.5]{ShiWeiXu2026}; replacing the Veronese curve with the principal curve extends that construction to every complex simple Lie algebra. The statements in this section are transfer and corollary results conditional on the input spherical metric, not new general PDE existence theorems.

In our Toda normalization, a spherical metric has Gaussian curvature $4$. A curvature-one spherical metric $h$ as in Theorem~\ref{thm:mondello-panov} below is therefore rescaled to $\omega=\frac14h$. This constant rescaling leaves its cone divisor, developing-map monodromy, and cone angles unchanged.

\subsection{Cone divisors and reduced solutions}

Let $X$ be a compact Riemann surface, let $P_1,\ldots,P_m\in X$ be distinct, and define $X^\circ=X\setminus\{P_1,\ldots,P_m\}$. A spherical metric $\omega$ on $X^\circ$ is said to represent the real divisor
\[
D=\sum_{\nu=1}^m\gamma_\nu[P_\nu],
\qquad \gamma_\nu>-1,
\]
if, in a coordinate $z$ centered at $P_\nu$,
\[
\omega=\frac{\sqrt{-1}}2e^{u}dz\wedge d\bar z,
\qquad
u=2\gamma_\nu\log|z|+O(1).
\]
Thus the cone angle at $P_\nu$ is $2\pi(1+\gamma_\nu)$. Similarly, an $n$-tuple $\boldsymbol\varphi=(\varphi_1,\ldots,\varphi_n)$ represents the divisor tuple $(D_1,\ldots,D_n)$ if, for every $i$ and $\nu$, the local exponent of $\varphi_i$ at $P_\nu$ is the coefficient of $P_\nu$ in $D_i$.

For type $A_n$, Mu--Shi--Sun--Xu \mbox{\cite[Theorem~1.2]{MuShiSunXu2024}} describe local finite-energy solutions with prescribed singular sources by normalized holomorphic data and identify their cone angles. The next lemma concerns preservation of these orders under the global transformations defining a reduced family.
\begin{lem}\label{lem:cone-orders}
Let $\mathfrak g$ be a complex simple Lie algebra, let $\omega$ be a cone spherical metric on $X^\circ$ representing $D$, let $v:\widetilde{X^\circ}\to\PP^1$ be a unitary developing curve, and let $\delta\in\Delta_v^{\mathfrak g}$. The Toda forms induced by
\[
f_\delta=\delta\cdot r_{\mathfrak g}\circ v
\]
have the same cone divisor $D$ in every component.
\end{lem}

\begin{proof}
For $\delta=e$, Theorem~\ref{thm:principal-reduced-solution} gives
\[
\varphi_i=m_i\omega.
\]
Multiplication by the positive constant $m_i$ does not alter the logarithmic order at a puncture.

For any fixed $\delta\in\Delta_v^{\mathfrak g}$, let $\delta_i$ denote its induced projective transformation of $\PP(V_i)$. Since $\PP(V_i)$ is compact and both $\varphi_{\mathrm{FS}}$ and $\delta_i^*\varphi_{\mathrm{FS}}$ are positive forms, there are constants $c_i,C_i>0$ such that
\[
c_i\varphi_{\mathrm{FS}}
\leq\delta_i^*\varphi_{\mathrm{FS}}
\leq C_i\varphi_{\mathrm{FS}}.
\]
The equivariance of $\pi_i$ gives
\[
\pi_i\circ f_\delta
=\delta_i\circ\pi_i\circ r_{\mathfrak g}\circ v.
\]
After pulling back the preceding inequalities, the $i$-th transformed Toda form $\varphi_i^\delta$ therefore satisfies
\[
c_i m_i\omega\leq\varphi_i^\delta\leq C_i m_i\omega.
\]
The ratio of the local densities of $\varphi_i^\delta$ and $\omega$ is thus bounded above and below by positive constants. Taking logarithms changes the local expression only by an $O(1)$ term, so the power $2\gamma_\nu$ of $|z|$ is unchanged for every $i$ and $\nu$.
\end{proof}

\begin{thm}
\label{thm:cone-transfer}
Let $\mathfrak g$ be a complex simple Lie algebra of rank $n$. Suppose that $X$ carries a cone spherical metric $\omega$ representing
\[
D=\sum_{\nu=1}^m\gamma_\nu[P_\nu],
\qquad \gamma_\nu>-1.
\]
Then the $\mathfrak g$-Toda system on $X$ with the diagonal cone-divisor tuple
\[
(D,\ldots,D)
\]
has a possibly zero-dimensional family of reduced solutions.
\end{thm}

\begin{proof}
Restricted to $X^\circ$, $\omega$ is a smooth positive $\operatorname{SU}(2)$ Toda solution and its developing map is locally univalent, so Theorems~\ref{thm:principal-reduced-solution} and~\ref{thm:classification-reduced-solutions} apply. Lemma~\ref{lem:cone-orders} shows that every member has divisor $D$ in each component. The dimension assertions are precisely those of Theorem~\ref{thm:classification-reduced-solutions}.
\end{proof}

\begin{rmk}
The solutions in Theorem~\ref{thm:cone-transfer} satisfy the singular Toda system in the sense of distributions, with source matrix
\[
\Gamma=(\gamma_{\nu,i})_{m\times n},
\qquad
\gamma_{\nu,i}=\gamma_\nu.
\]
Indeed, fix a coordinate disk $\mathbb D$ centered at $P_\nu$. Lemma~\ref{lem:cone-orders} gives
\[
u_i=2\gamma_\nu\log|z|+b_i,
\qquad b_i\in L^\infty(\mathbb D).
\]
Since $\gamma_\nu>-1$, we may choose $1<q<2$ such that $q\gamma_\nu>-1$. Then $e^{u_j}=O(|z|^{2\gamma_\nu})\in L^q(\mathbb D)$, and the punctured Toda equation becomes
\[
\partial_z\partial_{\bar z}b_i
=F_i:=-\sum_{j=1}^n a_{ji}e^{u_j}
\qquad\text{on }\mathbb D\setminus\{0\}.
\]

Choose $w_i\in W^{2,q}(\mathbb D)$ solving $\partial_z\partial_{\bar z}w_i=F_i$ on $\mathbb D$. By Sobolev embedding, $w_i$ is continuous and bounded near $0$. Thus $b_i-w_i$ is bounded and harmonic on a punctured neighborhood of $0$, so its singularity is removable. Consequently, $b_i$
extends continuously across $0$ and satisfies $\partial_z\partial_{\bar z}b_i=F_i$ distributionally there. Using
\[
\partial_z\partial_{\bar z}\log|z|
=\frac{\pi}{2}\delta_0,
\]
we obtain
\[
\partial_z\partial_{\bar z}u_i
+\sum_{j=1}^n a_{ji}e^{u_j}
=\pi\gamma_\nu\delta_0.
\]
Hence the logarithmic term accounts for the entire source at $P_\nu$, and all columns of $\Gamma$ agree. Moreover, $e^{u_i}=|z|^{2\gamma_\nu}e^{b_i}$ has a positive continuous factor after extracting the prescribed cone order.
\end{rmk}

\subsection{Positive-genus examples}

We recall the following existence theorem for spherical surfaces.

\begin{thm}[Mondello--Panov {\cite[Theorem~A]{MondelloPanov2019}}]
\label{thm:mondello-panov}
Let $g,m>0$ be integers and let $\beta_1,\ldots,\beta_m>0$ satisfy
\[
\sum_{\nu=1}^m\beta_\nu>2g-2+m.
\]
Then there exist a compact Riemann surface $X$ of genus $g$, distinct points $P_1,\ldots,P_m\in X$, and a curvature-one spherical metric $h$ representing
\[
D=\sum_{\nu=1}^m(\beta_\nu-1)[P_\nu].
\]
Consequently, $\omega:=\frac14h$ is a curvature-four spherical metric in our Toda normalization and represents the same divisor $D$.
\end{thm}

\begin{cor}
\label{cor:positive-genus-applications}
Under the hypotheses of Theorem~\ref{thm:mondello-panov}, for every complex simple Lie algebra $\mathfrak g$ of rank $n$, the $\mathfrak g$-Toda system on $X$ with the diagonal cone-divisor tuple
\[
\underbrace{(D,\ldots,D)}_{n\ \mathrm{components}}
\]
admits the possibly zero-dimensional reduced family described in Theorem~\ref{thm:cone-transfer}. Its dimension is determined by the monodromy closure. In particular, it admits the explicit solution
\[
(m_1\omega,\ldots,m_n\omega).
\]
\end{cor}

\begin{proof}
Apply Theorem~\ref{thm:cone-transfer} to the rescaled spherical metric $\omega=\frac14h$ supplied by Theorem~\ref{thm:mondello-panov}.
\end{proof}

\begin{example}
    For a concrete example, let $E_\lambda$ be the smooth projective completion of
\[
y^2=x(x-1)(x-\lambda),\qquad \lambda\in\CC\setminus\{0,1\},
\]
and let $f=x:E_\lambda\to\PP^1$ be the degree-two map. Its four simple ramification points $P_1,\ldots,P_4$ give a curvature-four spherical metric $\omega=f^*\varphi_{\mathrm{FS}}$ with divisor $D=P_1+\cdots+P_4$, so all four cone angles are $4\pi$. Its developing monodromy is trivial. Hence Theorem~\ref{thm:cone-transfer} gives, for every simple $\mathfrak g$, a family with
\[
\Delta_v^{\mathfrak g}=AN,\qquad
\dim_{\RR}\Delta_v^{\mathfrak g}=\dim_{\CC}\mathfrak g.
\]
All members have the same source coefficients $\gamma_{\nu,i}=1$. In the usual mean-field normalization their mass parameters are
\[
\rho_i:=4\int_{E_\lambda}\varphi_i^\delta
=8\pi m_i\in4\pi\ZZ_{>0},\qquad \delta\in AN.
\]
Indeed, $\int_{E_\lambda}\omega=2\pi$, and integration of the Toda equations fixes each component's area throughout this family. Here $m_i=\deg(\pi_i\circ r_{\mathfrak g})$ is a positive integer; writing $U_i=u_i+\log4$ gives the usual equation $\Delta U_i+\sum_j a_{ji}e^{U_j}=4\pi\sum_\nu\gamma_{\nu,i}\delta_{P_\nu}$ on a flat torus.
\end{example}

\begin{rmk}
\label{rmk:comparison-existence}
Corollary~\ref{cor:positive-genus-applications} produces reduced families beyond the hypotheses of the noncritical existence theorems discussed below. This extends the type-$A$ comparison of Shi, Wei, and Xu \cite[Remark~4.3]{ShiWeiXu2026} to every complex simple Lie algebra. The distinction concerns both solvable singularity data and the explicit parametrization of the resulting solutions.

For the types $A_n,B_n,C_n$ covered by Lin--Yang--Zhong \cite[Theorem~1.9]{LinYangZhong2020}, existence is obtained under the condition $\rho_i\notin\Gamma_i$ for every $i$, where their critical sets satisfy $4\pi\ZZ_{>0}\subset\Gamma_i$. Thus the displayed family lies outside those hypotheses; it also lies outside the parameter range of their degree-theoretic result \cite[Theorem~1.8]{LinYangZhong2020}. In type $A_2$, it has $(\rho_1,\rho_2)=(16\pi,16\pi)$ and an eight-dimensional reduced family. The later noncritical existence statement recorded by Chen--Lin \cite[Theorem~C]{ChenLin2023} requires $\sum_\nu\gamma_{\nu,1}\not\equiv\sum_\nu\gamma_{\nu,2}\pmod3$ for integral source coefficients; our diagonal data violate this condition as well.

Critical-parameter solutions in type $A_2$ have also been studied by Chen--Lin \cite{ChenLin2024}, and geometric type-$A$ constructions beyond noncritical existence criteria occur in Mu--Shi--Xu \cite[\S5]{MuShiXu2025}. The new phenomenon established here is the uniform realization of these explicit reduced families in every simple Lie type, together with their intrinsic centralizer parametrization and monodromy-dependent dimensions, including types $D_n,E_6,E_7,E_8,F_4,G_2$ absent from the cited Lin--Yang--Zhong existence theorem. These conclusions provide additional geometric information beyond the existence assertions of the cited analytic results.
\end{rmk}

\section{The diagonal-divisor converse question}\label{sec:counterexample}

Let $\mathcal R_D$ denote the reduced solutions generated by cone spherical metrics representing $D$. Theorem~\ref{thm:cone-transfer} proves the inclusion
\[
\mathcal R_D\subseteq\mathcal T_D:=\{\text{solutions with tuple }(D,\ldots,D)\}.
\]
The converse is false under the present hypotheses. The diagonal divisor condition records only the orders of the fundamental Pl\"ucker forms. Reducedness is a global condition requiring the associated unitary integral curve to be of the form $\delta\cdot r_{\mathfrak g}\circ v$ for one spherical developing pair $v$ and one constant $\delta\in G_{\mathrm{ad}}$.

Let $X=\PP^1$, let $P=\{0,1,-1,\infty\}$, and on $X\setminus P$ consider the holomorphic curve $f=[P_0:P_1:P_2]$ with
\[
P_0=1+2z-5z^2,\qquad
P_1=z^2-2z^3-z^4,\qquad
P_2=z^4+\frac25z^5-\frac15z^6.
\]
Its Wronskian is
\[
\det(f,f',f'')=16z^3(1-z^2)^3.
\]
The local vanishing sequence is $(0,2,4)$ at each point of $P$. Hence both fundamental Toda forms have cone coefficient $1$ at each point, and the associated $\operatorname{SU}(3)$ Toda solution represents
\[
(D,D),\qquad D=[0]+[1]+[-1]+[\infty].
\]
The solution is not reduced. Indeed, a reduced $A_2$ curve is a projective transform of a Veronese curve and is therefore contained in a projective conic. The above curve is not contained in any conic: if a quadratic relation among $P_0^2,P_0P_1,P_0P_2,P_1^2,P_1P_2,P_2^2$ is written with coefficients $A,\ldots,F$, comparison of the coefficients of degrees $12,10,8,7,6,4$ successively forces $F=E=C=D=B=A=0$. Thus the diagonal cone-divisor tuple does not imply reducedness, already for $\mathfrak g=\mathfrak{sl}_3$.

\begin{rmk}
What extra hypothesis would be needed for a converse? A converse would require a global rigidity assumption classifying unitary integral curves with equal fundamental ramification divisors. In type $A_2$, one possible sufficient hypothesis is that the projective image is a conic; then the curve is principal after reparametrization, and Theorem~\ref{thm:classification-reduced-solutions} applies. Equal cone orders alone do not impose this hypothesis.
\end{rmk}

\subsection*{Acknowledgements}
The detailed type-by-type computations for all simple Lie algebras in Section~\ref{sec:classification} were carried out with the assistance of GPT-6. The counterexample in Section~\ref{sec:counterexample} was also generated with its assistance. The authors take full responsibility for the content of this manuscript.
Y.S. is supported in part by the National Natural Science Foundation of China (Grant No. 11931009). B.X. is supported in part by the Project of Stable Support for Youth Team in Basic Research Field, CAS (Grant No. YSBR-001) and NSFC (Grant Nos. 12271495).

\end{document}